\documentclass[11pt]{article}

\usepackage[final]{acl}

\usepackage[T1]{fontenc}
\usepackage[utf8]{inputenc}
\usepackage{times}
\usepackage{microtype}
\usepackage{amsmath}
\usepackage{amssymb}
\usepackage{amsthm}
\usepackage{graphicx}
\usepackage{booktabs}
\usepackage{colortbl}
\usepackage{listings}
\usepackage{multirow}
\usepackage{siunitx}
\usepackage{subcaption}
\usepackage[most]{tcolorbox}
\usepackage{float}
\usepackage[english]{babel}

\tcbset{
  fdpromptbox/.style={
    enhanced,
    width=\textwidth,
    colback=black!1!white,
    colframe=black!45,
    boxrule=0.45pt,
    arc=1.4mm,
    outer arc=1.4mm,
    left=4mm,
    right=4mm,
    top=2.8mm,
    bottom=2.8mm,
    colbacktitle=black!8!white,
    coltitle=black,
    fonttitle=\bfseries\small,
    borderline west={1pt}{0pt}{black!45},
    attach boxed title to top left={xshift=3mm,yshift*=-\tcboxedtitleheight/2},
    boxed title style={
      colback=black!8!white,
      colframe=black!45,
      boxrule=0.45pt,
      arc=1mm,
      left=2mm,
      right=2mm,
      top=0.8mm,
      bottom=0.8mm
    }
  }
}
\newtheorem{proposition}{Proposition}
\newtheorem{lemma}{Lemma}
\theoremstyle{definition}
\newtheorem{definition}{Definition}
\newtheorem{assumption}{Assumption}

\title{FD-RAG: Federated Dual-System Retrieval-Augmented Generation}

\author{Tianhao Gao \and Kai Yang\thanks{Corresponding author.} \and Yiyang Li \\
  School of Computer Science and Technology, Tongji University \\
  thgao@tongji.edu.cn, kaiyang@tongji.edu.cn}

\begin{document}
\maketitle

\begin{abstract}
Retrieval-augmented generation (RAG) has emerged as a paradigm for grounding large language models in external knowledge,
yet most existing RAG systems assume centralized knowledge access and ample computation.
These assumptions break down in edge environments, where knowledge is fragmented across devices,
raw data cannot be shared, and repeated LLM calls are prohibitively expensive.
We propose FD-RAG, a federated dual-system RAG framework
that decouples lightweight memory access from on-demand LLM reasoning
for decentralized deployment.
Specifically, FD-RAG learns semantic-aware adaptive hypergraphs
over local corpora and distills them into compact QA memories.
At inference time, it answers well-covered queries via direct memory matching
and invokes LLM-based reasoning only when necessary,
while tracing retrieved memories to hypergraph-grounded evidence.
To mitigate cross-device knowledge fragmentation,
FD-RAG aggregates anonymized memories across devices without exposing raw documents.
Experiments on QA benchmarks show that FD-RAG improves accuracy by up to 7.8\% while reducing latency by 8.4$\times$
compared with strong local and federated baselines.
We also provide theoretical analysis establishing an $\mathcal{O}(1/\epsilon^{2})$ convergence rate for the proposed hypergraph learning,
supporting its tractable deployment in edge settings.
\end{abstract}

\section{Introduction}
\label{sec:intro}
Retrieval-Augmented Generation (RAG)~\citep{lewis2020retrieval} has emerged as a standard 
paradigm for grounding large language models (LLMs) in external knowledge, 
achieving strong performance on knowledge-intensive tasks. However, 
its effectiveness implicitly relies on a centralized setting where 
both knowledge and computation are readily accessible~\citep{oche2025systematic}. 
This assumption rarely holds in practice. 
In real-world domains such as healthcare, finance, and law, knowledge is inherently 
distributed across institutions and edge devices, data sharing is 
restricted by privacy and regulation~\citep{xu2025distributed}, 
and inference must operate under strict computation and latency constraints~\citep{liu2025nested}. These challenges call for 
RAG to operate directly on edge devices while respecting data locality.

Extending RAG to such environments introduces fundamental challenges. 
Existing approaches~\citep{luo2025hypergraphrag,edge2024local} rely on 
LLMs not only for reasoning but also for knowledge construction 
and query understanding. This design becomes fragile with small 
language models (SLMs), common in edge settings: knowledge 
construction degrades~\citep{fan2025minirag}, complex queries become 
less reliable, and repeated invocations incur prohibitive latency~\citep{liu2024mobilellm}. 
Moreover, most RAG frameworks assume a unified knowledge repository~\citep{gao2023retrieval}. 
In decentralized environments, knowledge is fragmented across devices, 
leaving each node with incomplete coverage, especially for queries 
requiring multi-source evidence~\citep{chakraborty-etal-2025-federated}. 
As a result, existing systems struggle to achieve both efficiency 
and completeness in edge scenarios.

We argue that these limitations stem from treating retrieval and reasoning 
as a monolithic process mediated by language models, 
without accounting for their distinct computational roles. 
Inspired by Dual-Process Theory~\citep{kahneman2003maps}, 
which distinguishes fast, memory-based responses (System 1) from slower, 
deliberative reasoning (System 2), we propose to decouple these two modes in RAG.
This leads to \textbf{Federated Dual-System RAG (FD-RAG)}, a unified framework that separates lightweight memory access from selective reasoning, enabling efficient and expressive knowledge utilization under resource constraints.

At the knowledge construction level, FD-RAG constructs a lightweight 
yet expressive knowledge structure via semantic-aware hypergraph learning, 
capturing higher-order relations without expensive extraction. 
This structure is further distilled into question--answer (QA) memories, 
serving as a compact interface for efficient access. 
At the inference level, FD-RAG introduces two complementary modules: a \textit{Memorizer}, 
which directly resolves queries well covered by QA memory through efficient matching, 
and a \textit{Cognizer}, which selectively invokes LLM-based reasoning 
over hypergraph-grounded evidence for more complex queries. 
To address knowledge fragmentation, 
FD-RAG further incorporates a federated memory aggregation mechanism, 
enabling multiple devices to collaboratively improve coverage without sharing raw data.
Our main contributions are as follows:
\begin{itemize}
    \item We propose FD-RAG, a dual-system RAG framework for edge environments that decouples memory-based retrieval and selective LLM reasoning, enabling collaborative knowledge utilization across distributed devices.
    \item We introduce a semantic-aware hypergraph learning approach for constructing lightweight yet expressive knowledge structures, and derive compact QA memories for efficient inference under resource constraints.
    \item Experiments on standard QA benchmarks show that FD-RAG improves accuracy by up to 7.8\% while reducing latency by a factor of $8.4\times$. We further prove that the proposed hypergraph learning procedure achieves an $\mathcal{O}(1/\epsilon^{2})$ convergence rate, providing theoretical guarantees for its efficiency and stability.
\end{itemize}

\section{Related Work}
\label{sec:related_work}
% !TEX root = ../acl_lualatex.tex

\paragraph{Retrieval-Augmented Generation.}
Recent RAG research has moved beyond the retrieve-then-generate pipeline to better support complex question answering, primarily through iterative retrieval and structured knowledge modeling \citep{jin2024long, li2024long}. Iterative retrieval methods improve evidence coverage by interleaving retrieval, generation, and query reformulation, but their repeated reliance on large language models often incurs substantial latency and computational cost \citep{liu2025hoprag, asai2023self}. Another line of work organizes evidence with explicit structures, including trees~\citep{sarthi2024raptor}, graphs~\citep{gutierrez2025from, li2024graphreader}, and more recently hypergraphs~\citep{luo2025hypergraphrag}. Compared with trees and graphs, hypergraphs can capture higher-order semantic relations beyond pairwise dependencies, making them particularly appealing for complex reasoning. However, existing structured RAG methods typically rely on expensive knowledge construction, large context windows, or strong semantic reasoning capabilities, which limits their suitability for small models in resource-constrained edge settings. Although efficiency-oriented methods such as EfficientRAG reduce part of the computational burden \citep{zhuang2024efficientrag}, \textit{how to preserve the benefits of structured knowledge while enabling efficient, low-latency inference on edge devices remains largely underexplored.}

\paragraph{Federated Retrieval-Augmented Generation.}
Federated RAG extends RAG to settings where knowledge is inherently distributed across silos and cannot be centralized due to privacy or regulatory constraints. Existing work mainly explores two complementary directions. One line of research focuses on federated retrieval, enabling cross-silo access via routing or aggregation mechanisms~\citep{wang2024feb4rag,guerraoui2025ragroute,xu2024cfedrag,zhao2024frag}, but still relies on centralized coordination, limiting system autonomy and flexibility. Another line integrates RAG with federated learning, jointly optimizing model parameters across clients~\citep{he2025pfedrag,liang2026fedmosaic,fajardo2025fedragframework,ma2025cellular}, yet incurs substantial communication and computation overhead due to frequent parameter exchange. \textit{Despite these advances, prior work largely overlooks collaboration at the knowledge level, hindering support for dynamic, unstructured, and semantically rich RAG scenarios in fully decentralized environments.}

\begin{figure*}[t]
    \centering
    \includegraphics[width=\textwidth]{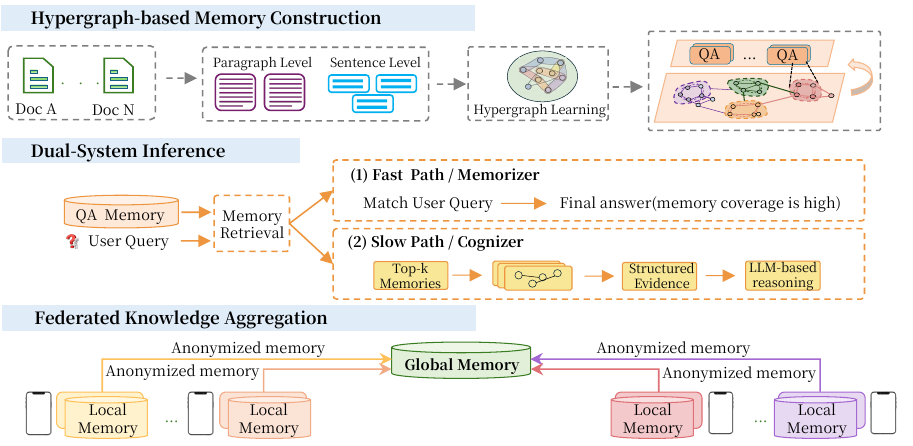}
    \caption{Overview of FD-RAG. We organize each local corpus into a semantic hypergraph and convert hyperedges into grounded QA memories. At inference time, the \textit{Memorizer} answers high-confidence queries via direct memory matching, while the \textit{Cognizer} retrieves structured evidence from supporting hyperedges and performs LLM-based reasoning. Local memories can be anonymized and aggregated across devices under federated constraints.}
    \label{fig:overview}
\end{figure*}

\section{Preliminaries}
\label{sec:preliminaries}
% !TEX root = ../acl_lualatex.tex
\subsection{Hypergraph}

A hypergraph is defined as $\mathcal{H} = (\mathcal{V}, \mathcal{E})$, where $\mathcal{V} = \{v_1, \ldots, v_N\}$ denotes the set of nodes, and $\mathcal{E} = \{e_1, \ldots, e_M\}$ denotes the set of hyperedges, each connecting a subset of nodes in $\mathcal{V}$. 

The structure of the hypergraph is represented by an incidence matrix $H \in \mathbb{R}^{N \times M}$, where $H_{nm}$ indicates the membership of node $v_n$ in hyperedge $e_m$:
\begin{equation}
H_{nm} =
\begin{cases}
1, & \text{if } v_n \in e_m, \\
0, & \text{otherwise.}
\end{cases}
\end{equation}

\subsection{Problem Formulation}

Let $\mathcal{Z} = \{Z_k\}_{k=1}^K$ be a set of edge devices, where each device $Z_k$ holds a local corpus $\mathcal{T}_k$ and a collection of user queries $\mathcal{Q}_k = \{q_{kl}\}_{l=1}^{N_k}$. We aim to build a federated RAG system $\mathcal{M}$ that enables each device to produce accurate answers without exposing its local data or queries. Specifically, the answer to query $q_{kl}$ is generated as:
\begin{equation}
    a_{kl} = \mathcal{M}_k\!\left(q_{kl};\, \mathcal{T}_k,\, \mathcal{G}\right),
    \quad \mathcal{G} = \bigcup_{k=1}^{K} \Gamma_k,
\end{equation}
where $\Gamma_k$ is an anonymized, shareable knowledge summary distilled from $\mathcal{T}_k$, and $\mathcal{G}$ aggregates cross-device knowledge without centralizing raw data.

We seek Pareto-optimal solutions that maximize answer accuracy $\mathcal{U} = \sum_{k,l} \mathrm{Acc}(a_{kl})$ and minimize end-to-end latency $\mathcal{L}$ under knowledge-sharing constraints:
\begin{equation}
\begin{aligned}
    \max_{\mathcal{M},\,\{\Gamma_k\}} \quad
    & \bigl[\mathcal{U},\; {-\mathcal{L}(\mathcal{M})}\bigr]^\top \\
    \mathrm{s.t.} \quad
    & \Gamma_k = \Phi(\mathcal{T}_k), \quad \forall\, k \in [K],
\end{aligned}
\end{equation}
where $\Phi(\cdot)$ denotes local knowledge distillation. This formulation captures the accuracy--latency trade-off while restricting federation to distilled knowledge. The solutions lie on the Pareto frontier~\citep{hochman1969pareto}, enabling principled trade-offs in resource-constrained edge environments.

\section{Federated Dual-System RAG}
\label{sec:method}
% !TEX root = ../acl_lualatex.tex

In this section, we present FD-RAG, which consists of three components: 
Hypergraph-based Memory Construction (Section~\ref{sec:hypergraph}), 
Dual-System Inference (Section~\ref{sec:dual-system}), and Federated Knowledge 
Aggregation (Section~\ref{sec:federated}), as illustrated in Figure~\ref{fig:overview}.
Concretely, FD-RAG first constructs a semantic hypergraph over the local corpus and 
distills it into a set of hyperedge-grounded QA memories. During inference, the 
\textit{Memorizer} resolves queries with sufficient memory coverage via direct matching, 
while the \textit{Cognizer} handles harder queries by localizing relevant hyperedges and 
invoking LLM-based reasoning over the structured evidence. Cross-device knowledge 
gaps are addressed by sharing anonymized memories under federated constraints. 
We describe each component in detail below.

\subsection{Hypergraph-based Memory Construction}
\label{sec:hypergraph}

\paragraph{Semantic-Aware Hypergraph Learning.}
To capture higher-order semantic relations across text units without 
relying on costly extraction procedures, we learn the corpus structure 
directly from dense semantic representations.

A central design consideration is representational granularity: 
sentence-level units provide precise, localized anchors suitable for fact-level memory, 
while paragraph-level units preserve broader discourse context necessary for 
multi-hop and compositional reasoning. To capture both, we segment the 
local corpus $\mathcal{T}$ into paragraph-level units 
$P = \{p_i\}_{i=1}^{I}$ and sentence-level units 
$S = \{s_{j}\}_{j=1}^{J}$, 
and encode them with a pre-trained dense encoder 
(e.g., BGE-M3~\citep{chen2024bge}), yielding embedding matrices 
$E_P \in \mathbb{R}^{I \times D}$ and $E_S \in \mathbb{R}^{J \times D}$, 
where $D$ denotes the embedding dimension.

For each granularity $t \in \{p, s\}$, let 
$X^t \in \mathbb{R}^{N_t \times D}$ denote the node embedding matrix, 
where $X^p = E_P$ and $X^s = E_S$. 
We introduce $M^t$ learnable hyperedges and optimize a soft incidence 
matrix $\hat{H}^t \in [0,1]^{N_t \times M^t}$, where $\hat{H}_{nm}^t$ 
measures the association degree between node $n$ and hyperedge $m$. 
Each row is constrained to a probability simplex~\citep{boyd2004convex} to ensure interpretable membership:

\begin{equation}
    \sum_{m=1}^{M^t} \hat{H}_{nm}^t = 1, 
    \qquad \hat{H}_{nm}^t \ge 0.
\end{equation}

We then sparsify $\hat{H}^t$ by retaining only assignments above a 
threshold $\mu$, yielding a compact incidence matrix $H^t$. 
The prototype of each hyperedge $e_m^t$ is computed as the weighted mean 
of its incident node embeddings:
\begin{equation}
    e_m^t 
    = \frac{\sum_{n=1}^{N_t} H_{nm}^t\, x_n^t}
           {\sum_{n=1}^{N_t} H_{nm}^t} 
    \in \mathbb{R}^{D},
\end{equation}
where $x_n^t$ denotes the embedding of node $n$ at granularity $t$.

We optimize the hyperedge assignments via two complementary 
objectives~\citep{shang2024ada}. 
The intra-hyperedge term enforces semantic compactness within each 
hyperedge:
\begin{equation}
    L_{\text{intra}}^t
    =
    \frac{1}{M^t}
    \sum_{m=1}^{M^t}
    \frac{1}{|\mathcal{N}(e_m^t)|}
    \sum_{x_n^t \in \mathcal{N}(e_m^t)}
    \|x_n^t - e_m^t\|_2,
\end{equation}
where $\mathcal{N}(e_m^t)$ denotes the node set incident to $e_m^t$. 
The inter-hyperedge term regulates the global geometry of hyperedge 
prototypes: rather than enforcing indiscriminate separation, it attracts 
semantically similar hyperedges while repelling dissimilar ones:
\begin{equation}
% \resizebox{0.3\textwidth}{!}{$
    \begin{aligned}
    L_{\text{inter}}^t
    &=
    \frac{1}{(M^t)^2}
    \sum_{i=1}^{M^t}
    \sum_{j=1}^{M^t}
    \Bigl(\rho_{ij} \lVert e_i^t \!-\! e_j^t \rVert_2 \\
    &\qquad
        +
        (1-\rho_{ij})
        \max\bigl(\gamma - \lVert e_i^t - e_j^t \rVert_2,\, 0\bigr)
    \Bigr),
    \end{aligned}
\end{equation}
where $\rho_{ij}$ denotes the cosine similarity~\citep{salton1989automatic} between hyperedge 
prototypes $e_i^t$ and $e_j^t$, and $\gamma$ is a margin hyperparameter. 
The overall objective balances local compactness and global discrimination:
\begin{equation}
    L_{\text{total}} 
    = 
    \sum_{t \in \{p,s\}}
    \left(
        \lambda\, L_{\text{intra}}^t + (1-\lambda)\, L_{\text{inter}}^t
    \right).
\end{equation}
Here, $\lambda \in [0,1]$ balances the intra-hyperedge and inter-hyperedge terms.

After optimization, paragraph-level hyperedges $\mathcal{E}^p$ and 
sentence-level hyperedges $\mathcal{E}^s$ jointly form the overall 
hyperedge set $\mathcal{E} = \mathcal{E}^p \cup \mathcal{E}^s$ with 
$M = |\mathcal{E}|$. Each hyperedge $e_m \in \mathcal{E}$ is associated with 
a set of text units $\mathcal{C}(e_m)$ from the original corpus, 
enabling the downstream inference framework to retrieve the corresponding textual evidence.

\begin{proposition}\label{prop:main-convergence}\textit{(Stationarity of Hyperedge Assignment Learning.)}
Under $L$-smoothness of $L_{\textup{total}}$, the simplex-constrained
optimization of the soft incidence matrix $\hat{H}^t$ attains an
$\epsilon$-stationary point in $\mathcal{O}(1/\epsilon^2)$ iterations. The proof,
tailored to the probability-simplex constraints of our hyperedge assignment
formulation, is provided in Appendix~\ref{appendix:convergence-proof}.
\end{proposition}

\paragraph{QA Memory Construction.}
To expose the hypergraph as an efficient memory interface, 
we convert each hyperedge into a set of retrieval-oriented QA memories. 
For each hyperedge $e_m \in \mathcal{E}$, we first derive a  
typed fact set from its associated context $\mathcal{C}(e_m)$ using a lightweight
\textsc{spaCy}-based extractor~\citep{honnibal2020spacy}, 
providing a comprehensive factual basis for subsequent QA generation:
\begin{equation}
\label{eq:factset}
    \mathcal{F}_m 
    = \langle (u_1, v_1),\, (u_2, v_2),\, \ldots,\, 
               (u_{T_m}, v_{T_m}) \rangle,
\end{equation}
where $u_t$ denotes a textual fact span and $v_t$ its semantic type. 
This design delegates raw-corpus fact extraction to an efficient traditional 
NLP pipeline, preserving localized evidence while reducing the burden on the 
language model. Conditioned on $\mathcal{F}_m$ and $\mathcal{C}(e_m)$, we then 
use a local language model to synthesize hyperedge-grounded 
QA memory items (see Appendix~\ref{appendix:prompt} for the prompt template). This keeps the offline memory construction stage lightweight 
and well suited to edge deployment. Formally,
\begin{equation}
    \Gamma 
    = \bigcup_{m=1}^{M} 
      \left\{ \gamma_m^r = (q_m^r,\, a_m^r,\, \mathcal{S}_m^r) 
      \right\}_{r=1}^{R_m},
\end{equation}
where $q_m^r$ and $a_m^r$ denote the $r$-th question--answer pair 
grounded in hyperedge $e_m$, and 
$\mathcal{S}_m^r \subseteq \mathcal{E}$ is the supporting hyperedge set, 
with $e_m \in \mathcal{S}_m^r$ required.
Additional hyperedges are included when the question involves 
cross-hyperedge composition, naturally accommodating both fact-level 
memories ($|\mathcal{S}_m^r| = 1$) and multi-hop memories 
($|\mathcal{S}_m^r| > 1$).

The resulting memory $\Gamma$ serves as a structured access layer 
over the underlying semantic hypergraph rather 
than an independent synthetic QA pool. 
Each memory item maintains explicit links to its supporting hyperedges, 
preserving traceability to the original graph structure. 
During downstream inference, the model can answer directly from memory 
when the stored information is sufficient, or retrieve structured evidence 
from the hypergraph when more elaborate reasoning is required.

\subsection{Dual-System Inference Framework}
\label{sec:dual-system}

Queries exhibit heterogeneous reasoning demands: 
well-covered queries can be answered via direct memory matching, 
while complex ones require structured retrieval and LLM-based reasoning. 
Uniform LLM usage incurs unnecessary latency, 
whereas LLM-free methods sacrifice accuracy. 
To balance this trade-off, we propose a dual-system inference framework 
that routes queries based on their coverage under the QA memory $\Gamma$.

\paragraph{Unified Matching Score.}
To quantify memory coverage, we score each candidate
$\gamma_r = (q_r, a_r, \mathcal{S}_r) \in \Gamma$
against the user query $q$ by combining two complementary signals.
Dense semantic similarity alone may miss structural alignment
on entity-centric queries, while structural overlap alone is brittle
when surface forms vary.
Their combination yields a robust, low-cost coverage estimate:
\begin{equation}
\begin{aligned}
\operatorname{Score}(q,\gamma_r)
&= \alpha\, \operatorname{Sim}\!\bigl(f(q), f(q_r)\bigr) \\
&\quad + (1-\alpha)\, \operatorname{Cover}(q,\mathcal{S}_r)
\end{aligned}
\label{eq:score}
\end{equation}
where $f(\cdot)$ is a dense text encoder, $\alpha\!\in\![0,1]$
balances the two signals, and
\begin{equation}
\operatorname{Cover}(q,\mathcal{S}_r)
=
\frac{2\lvert\mathcal{A}_q \cap \mathcal{A}_r\rvert}
     {\lvert\mathcal{A}_q\rvert + \lvert\mathcal{A}_r\rvert}
\label{eq:cover}
\end{equation}
measures the Dice overlap between the query anchor set
$\mathcal{A}_q$ (salient named entities, noun phrases,
and typed concept spans extracted from $q$) and the support anchor set
$\mathcal{A}_r = \bigcup_{e \in \mathcal{S}_r}\mathcal{A}(e)$
derived from the typed facts of each supporting hyperedge.
 
\paragraph{Memorizer: Memory-Triggered Fast Thinking.}
For queries that are well-covered by memory, repeated LLM invocation
is wasteful: the answer is already latent in the stored QA pairs.
The \textit{Memorizer} therefore identifies the best-matched item
\begin{equation}
r^{*} = \arg\max_{r}\operatorname{Score}(q,\gamma_r),
\label{eq:argmax}
\end{equation}
and directly returns $a_{r^{*}}$ whenever
$\operatorname{Score}(q,\gamma_{r^{*}}) \ge \delta$,
bypassing LLM inference entirely.
This fast path eliminates the dominant source of latency
for the majority of queries while preserving answer fidelity,
as the retrieved answer is grounded in hyperedge-verified evidence
from the construction phase.
 
\paragraph{Cognizer: Hyperedge-Grounded Slow Thinking.}
When $\operatorname{Score}(q,\gamma_{r^{*}}) < \delta$,
direct memory matching is insufficient as the query may
require compositional reasoning across multiple evidence pieces
that no single memory item can cover.
Rather than falling back to full-corpus retrieval,
the \textit{Cognizer} reuses the memory layer as a
\emph{localization interface}:
the top-$K$ items $\mathcal{R}_q = \operatorname{TopK}(q,\Gamma)$
are retrieved, and their supporting hyperedges are aggregated as
\begin{equation}
\mathcal{E}_q = \bigcup_{\gamma_r \in \mathcal{R}_q} \mathcal{S}_r.
\end{equation}
For each $e \in \mathcal{E}_q$, the source context
$\mathcal{C}(e)$ and typed fact set $\mathcal{F}(e)$ are assembled
into structured evidence units
$z_e = \bigl(\mathcal{C}(e), \mathcal{F}(e)\bigr)$,
over which the LLM reasons using the prompt in
Appendix~\ref{appendix:prompt}:
\begin{equation}
\hat{a} = \mathrm{LLM}\bigl(q, \mathcal{Z}_q\bigr), \quad
\mathcal{Z}_q = \{ z_e \mid e \in \mathcal{E}_q \}.
\end{equation}
This design is critical for edge deployment:
by confining LLM reasoning to a small, hyperedge-selected evidence set
rather than the full corpus,
the slow path achieves targeted inference
without sacrificing compositional reasoning capability.

\subsection{Federated Knowledge Aggregation}
\label{sec:federated}
 
Local corpora on individual devices are inherently incomplete, and
evidence for a query may be distributed across devices rather than
available to any single one. To mitigate this knowledge fragmentation,
FD-RAG performs federation at the memory level: each device shares
QA memories distilled from its local hypergraph, while raw
corpora and full hypergraph structures remain on device.
 
\paragraph{Local Memory Export.}
Each device $k$ constructs a local semantic hypergraph $\mathcal{H}_k$
and QA memory $\Gamma_k$ from its corpus $\mathcal{T}_k$
(Section~\ref{sec:hypergraph}). For federation, the device uploads a
shareable memory view derived from $\Gamma_k$. In privacy-sensitive
settings, sensitive entities in the typed facts (Eq.~\ref{eq:factset})
are perturbed via randomized response~\citep{warner1965randomized} to
satisfy $\epsilon$-LDP~\citep{dwork2008differential}. To preserve
semantic utility, each entity $e$ is replaced with a surrogate sampled
from a candidate set $W$ ($|W|=c$) containing $e$ and $c-1$
semantically similar alternatives.

\begin{proposition}\label{prop:main-ldp}\textit{($\epsilon$-LDP Perturbation Mechanism.)}
Given a sensitive entity $e$ with candidate set
$W = \{e, w_1, \ldots, w_{c-1}\}$ of semantically similar alternatives,
the perturbation mechanism
\begin{equation*}
\Pr[e' \mid e] =
\begin{cases}
\dfrac{e^{\epsilon}}{e^{\epsilon} + c - 1}, & e' = e, \\[6pt]
\dfrac{1}{e^{\epsilon} + c - 1}, & e' \in W,\ e' \neq e,
\end{cases}
\end{equation*}
satisfies $\epsilon$-local differential privacy.
\textit{Proof in Appendix~\ref{appendix:ldp-proof}.}
\end{proposition}

This mechanism preserves utility while obfuscating device-specific identifiers. 
The anonymized memory is defined as

\begin{equation}
\widetilde{\Gamma}_k = \operatorname{Anonymize}(\Gamma_k).
\end{equation}

\paragraph{Global Memory Fusion.}
Given uploads $\{\widetilde{\Gamma}_k\}_{k=1}^{K}$, the server aggregates
them into a global memory bank:
\begin{equation}
\Gamma^{g} = \bigcup_{k=1}^{K} \widetilde{\Gamma}_k.
\end{equation}
Because different devices often contain complementary evidence,
$\Gamma^{g}$ extends the coverage of any single device. It improves the
\textit{Memorizer} by increasing the chance of high-confidence
matches, and assists the \textit{Cognizer} by providing cross-device
cues for localizing the hyperedges most likely to contain the required
evidence. Federation therefore serves primarily as knowledge-level
collaboration to reduce fragmentation, while privacy protection remains
an enabling safeguard for sensitive deployments.

\section{Experiments}
\label{sec:experiments}
% !TEX root = ../acl_lualatex.tex
\label{experiment}
We evaluate FD-RAG from three complementary perspectives.
\textbf{RQ1:} How does FD-RAG compare with representative baselines in overall answer quality and efficiency under both local and federated settings?
\textbf{RQ2:} How does the proposed Dual-System inference mechanism balance the Memorizer and the Cognizer, and does it provide a better accuracy--efficiency trade-off than always using either path alone?
\textbf{RQ3:} Are the main design choices in FD-RAG all necessary, and how does removing each component affect answer quality and latency?

\begin{table*}[t]
\centering
\footnotesize
\caption{Main results on three benchmarks under both local and federated settings. We report F1, accuracy (ACC), latency (Lat), and average LLM calls (Calls) per query. Best results in each column are shown in \textbf{bold}.}
\label{tab:main}
\setlength{\tabcolsep}{3.2pt}
\renewcommand{\arraystretch}{1.12}
\begin{tabular*}{\textwidth}{@{\extracolsep{\fill}}lcccccccccccc@{}}
\toprule
\multirow{2}{*}{\textbf{Method}} &
\multicolumn{4}{c}{\textbf{HotPotQA}} &
\multicolumn{4}{c}{\textbf{2WikiMQA}} &
\multicolumn{4}{c}{\textbf{MuSiQue}} \\
\cmidrule(lr){2-5} \cmidrule(lr){6-9} \cmidrule(lr){10-13}
& \textbf{F1} & \textbf{ACC} & \textbf{Lat} & \textbf{Calls}
& \textbf{F1} & \textbf{ACC} & \textbf{Lat} & \textbf{Calls}
& \textbf{F1} & \textbf{ACC} & \textbf{Lat} & \textbf{Calls} \\
\midrule

\rowcolor{black!6}
\multicolumn{13}{c}{\textbf{Local Setting}} \\

Vanilla RAG~\citep{lewis2020retrieval}
& 44.3 & 49.6 & 2.4s & 1.00
& 43.5 & 45.2 & 2.6s & 1.00
& 11.0 & 9.2 & 1.5s & 1.00 \\

IterDRAG~\citep{yueinference}
& 47.4 & 44.4 & 5.2s & 2.47
& 38.8 & 43.8 & 6.37s & 3.02
& 17.5 & 12.2 & 4.3s & 2.13 \\

LongRAG~\citep{jiang2024longrag}
& 57.3 & 53.0 & 10.4s & 6.23
& 58.1 & 54.0 & 12.3s & 7.35
& 35.5 & 31.0 & 9.2s & 5.48 \\

EfficientRAG~\citep{zhuang2024efficientrag}
& 52.1 & 50.0 & 3.3s & 2.65
& 46.3 & 44.2 & 3.9s & 3.22
& 18.3 & 17.0 & 3.6s & 2.86 \\

RAPTOR~\citep{sarthi2024raptor}
& 63.5 & 59.0 & 4.5s & 1.00
& 61.2 & 50.6 & 5.9s & 1.00
& 36.8 & 30.2 & 5.6s & 1.00 \\

GraphRAG~\citep{edge2024local}
& 55.1 & 50.8 & 6.9s & 1.00
& 58.9 & 52.7 & 7.4s & 1.00
& 34.2 & 29.8 & 7.1s & 1.00 \\

HyperGraphRAG~\citep{luo2025hypergraphrag}
& 60.7 & 58.4 & 4.2s & 1.00
& 61.8 & 57.1 & 4.9s & 1.00
& 37.4 & 34.8 & 4.1s & 1.00 \\

HippoRAG 2~\citep{gutierrez2025from}
& 63.2 & 60.4 & 3.8s & 1.00
& \textbf{63.5} & 57.9 & 3.9s & 1.00
& \textbf{40.3} & 36.2 & 3.4s & 1.00 \\

\textbf{FD-RAG-Local (Ours)}
& \textbf{73.4} & \textbf{68.2} & \textbf{0.45s} & \textbf{0.32}
& 63.1 & \textbf{59.2} & \textbf{0.62s} & \textbf{0.45}
& 39.1 & \textbf{38.0} & \textbf{0.54s} & \textbf{0.52} \\

\midrule
\rowcolor{black!6}
\multicolumn{13}{c}{\textbf{Federated Setting}} \\

\textsc{Local-RAG}
& 39.2 & 44.5 & 2.0s & 1.00
& 41.3 & 43.3 & 2.2s & 1.00
& 9.7 & 7.9 & 1.3s & 1.00 \\

C-FedRAG~\citep{xu2024cfedrag}
& 49.2 & 46.1 & 4.9s & 2.84
& 50.6 & 47.8 & 5.3s & 2.97
& 24.9 & 22.0 & 4.2s & 2.73 \\

RAGRoute~\citep{guerraoui2025ragroute}
& 53.4 & 50.9 & 3.0s & 1.31
& 55.9 & 52.8 & 3.2s & 1.38
& 28.7 & 26.0 & 2.7s & 1.24 \\

FD-RAG (w/o fusion)
& 65.0 & 61.2 & 1.42s & 0.80
& 60.3 & 56.7 & 1.58s & 0.86
& 35.6 & 32.8 & 1.46s & 0.82 \\

\textbf{FD-RAG (Ours)}
& \textbf{68.9} & \textbf{64.5} & \textbf{1.14s} & \textbf{0.62}
& \textbf{62.6} & \textbf{58.9} & \textbf{1.28s} & \textbf{0.68}
& \textbf{38.3} & \textbf{35.7} & \textbf{1.21s} & \textbf{0.65} \\

\bottomrule
\end{tabular*}
\end{table*}

\subsection{Experimental Setup}

\begin{table}[t]
\centering
\scriptsize
\caption{Dual-system inference decomposition on HotPotQA. Fast-path coverage is the fraction of queries resolved by the Memorizer; fast/slow ACC are conditional accuracies on each subset; Oracle is an upper bound only.}
\label{tab:dual}
\setlength{\tabcolsep}{0pt}
\renewcommand{\arraystretch}{1.1}
\begin{tabular*}{\columnwidth}{@{\extracolsep{\fill}}lccccc@{}}
\toprule
\textbf{Method} & \textbf{\shortstack{Fast\\Cover.}} & \textbf{\shortstack{Fast\\ACC}} & \textbf{\shortstack{Slow\\ACC}} & \textbf{\shortstack{Overall\\ACC}} & \textbf{\shortstack{Avg.\\Lat.}} \\
\midrule
FD-RAG (full)  & 68.0\% & 77.7 & 48.0 & 68.2 & 0.45s \\
Mem.-only      & 100.0\% & 38.2 & --- & 38.2 & 0.16s \\
Cog.-only      & 0.0\% & --- & 59.3 & 59.3 & 2.13s \\
Oracle (upper) & 80.0\% & --- & --- & 75.0 & 0.55s \\
\bottomrule
\end{tabular*}
\end{table}

\paragraph{Benchmarks.}
We follow the benchmark selection used in~\citet{jiang2024longrag}, evaluating our method on HotPotQA~\citep{yang2018hotpotqa}, 2WikiMQA~\citep{ho2020constructing}, and MuSiQue~\citep{trivedi2022musique}. For consistency in evaluation protocol, we adopt the experimental settings introduced in LongBench~\citep{bai2024longbench}. 
\paragraph{Dataset Construction.}
We evaluate all methods under two settings.
(1) Local setting: Each query is assigned to a home client whose local document store contains all of its gold supporting documents $D^{+}(q)$, so that no cross-silo retrieval is required to answer it. Since conventional RAG baselines are not designed with a decentralized protocol, evaluating them under the federated setting would introduce an inherent architectural disadvantage; we therefore restrict their evaluation to this setting to ensure a fair and controlled comparison.
(2) Federated setting: Following prior work~\citep{wang2024feb4rag}, we partition the corpus into disjoint subsets assigned to different clients, simulating a multi-silo environment. For each query $q$, we define it as \emph{local} if all documents in $D^{+}(q)$ reside on its home client, and \emph{cross-silo} otherwise. This protocol preserves the original QA pairs while introducing controlled distribution of supporting evidence across clients. Unless otherwise specified, all methods share the same document partition and query assignment on each dataset. Detailed dataset statistics are provided in Appendix~\ref{appendix:experiment-details}.

\begin{figure}[t]
\centering
\makebox[\columnwidth][c]{\includegraphics[width=0.90\columnwidth]{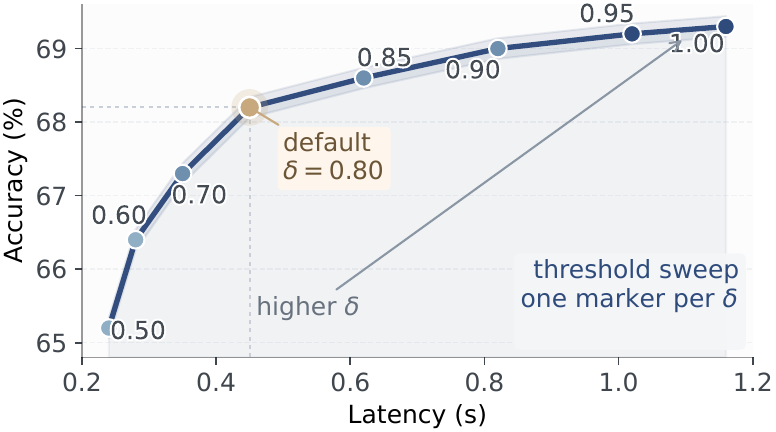}}
\caption{Pareto frontier of accuracy and latency on HotPotQA under varying confidence threshold $\delta$. Each marker corresponds to one threshold setting, and the highlighted marker denotes the default $\delta=0.8$.}
\label{fig:acc_latency_tradeoff}
\end{figure}

\begin{figure*}[!t]
\centering
\includegraphics[width=0.98\textwidth]{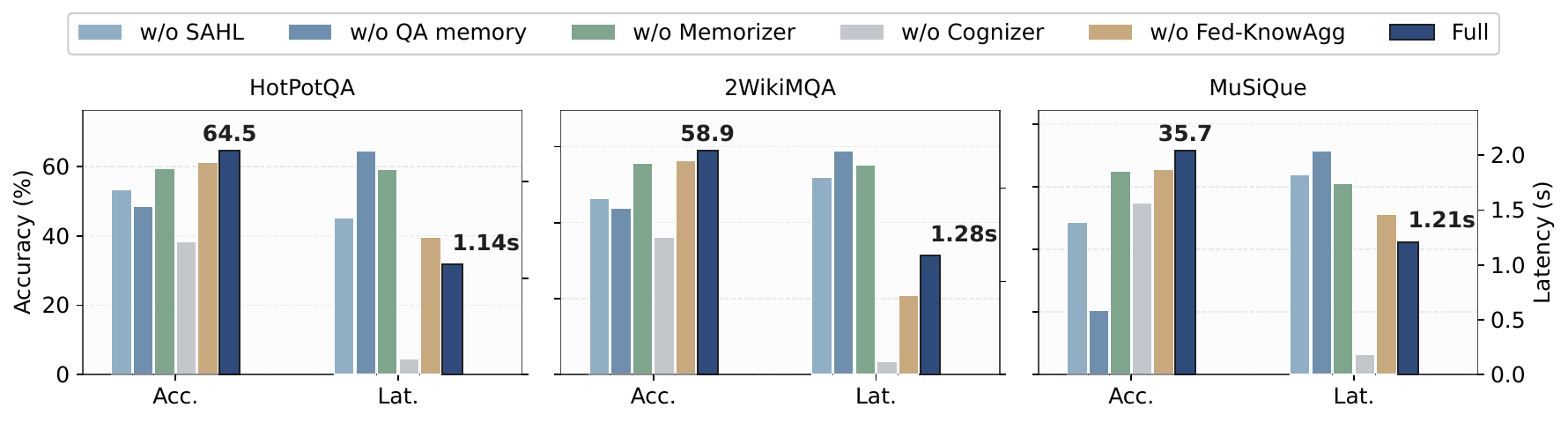}
\caption{Ablation study on HotPotQA, 2WikiMQA, and MuSiQue. Each panel reports accuracy (Acc.) and latency (Lat.) for the full model and five ablated variants.}
\label{fig:fd_rag_ablation}
\end{figure*}

\paragraph{Evaluation Metrics.}
We report accuracy (ACC), F1, and average end-to-end latency per query as primary metrics. Latency is measured on device from query receipt to final answer generation and averaged over the evaluation set. We further report the average number of LLM calls per query to characterize inference cost. In federated settings, metrics are computed per client and averaged across clients.
\paragraph{Baselines.}
We compare FD-RAG against three groups of baselines.
(1) \textbf{RAG baselines}: Vanilla RAG~\citep{lewis2020retrieval}, LongRAG~\citep{jiang2024longrag}, IterDRAG~\citep{yueinference}, EfficientRAG~\citep{zhuang2024efficientrag}, RAPTOR~\citep{sarthi2024raptor}, GraphRAG~\citep{edge2024local}, HippoRAG 2~\citep{gutierrez2025from}, and HyperGraphRAG~\citep{luo2025hypergraphrag}.
(2) \textbf{Baselines under the federated setting}: \textsc{Local-RAG}: a non-collaborative federated baseline where each client performs RAG using only its own local data, RAGRoute~\citep{guerraoui2025ragroute} and C-FedRAG~\citep{xu2024cfedrag}.
(3) \textbf{FD-RAG variants}: \textsc{FD-RAG-Local}, the single-device version without federated aggregation, and \textsc{FD-RAG (w/o fusion)}, which keeps the federated deployment protocol but removes the global memory fusion stage.
\paragraph{Implementation Details.}
We use $\mu=0.5$, $K=5$, $\lambda=0.6$, and $\delta=0.8$ across all datasets and clients. For privacy-preserving memory sharing, we set the local differential privacy budget to $\epsilon=1.0$ with candidate set size $c=5$. The semantic-aware hypergraph objective is optimized for 300 steps with a learning rate of 0.05, using a fixed hyperparameter configuration throughout. 
We adopt BGE-M3~\citep{chen2024bge} as the dense encoder and \texttt{Llama-3.1-8B}~\citep{grattafiori2024llama} (INT4 quantized) as the language model. All federated experiments use five clients. FD-RAG is deployed on an NVIDIA Jetson Orin Nano 8GB~\citep{nvidia_jetson_orin}, with end-to-end latency measured on device using batch size 1. Additional deployment details and fairness controls are provided in Appendix~\ref{app:implementation}.

\subsection{Main Results (RQ1)}

Table~\ref{tab:main} shows that \textsc{FD-RAG} consistently achieves a stronger 
accuracy--efficiency frontier than both conventional RAG baselines and federated competitors. 
In the local setting, \textsc{FD-RAG-Local} achieves the best ACC on all benchmarks 
and attains 68.2 on HotPotQA, outperforming HippoRAG~2 by 7.8\% ACC while reducing latency by 8.4$\times$. 
The gain is not merely due to using structured retrieval, since \textsc{FD-RAG-Local} also surpasses 
GraphRAG and HyperGraphRAG while keeping the average number of LLM calls below 0.6 across datasets. 
This suggests that the benefit comes from exposing the learned structure as 
QA memory and invoking LLM reasoning only when memory coverage is insufficient. 
In the federated setting, the drop of \textsc{Local-RAG} to 44.5 ACC on HotPotQA 
highlights the severity of cross-silo knowledge fragmentation. \textsc{FD-RAG} substantially 
recovers this gap, outperforming RAGRoute by 13.6 ACC points while remaining 2.6$\times$ faster, 
and it consistently improves over \textsc{FD-RAG}~(w/o fusion) across all benchmarks. 
Together, these results indicate that memory-level federation effectively 
expands coverage for distributed evidence while 
preserving the lightweight inference profile of the local model.

\subsection{Dual-System Inference Analysis (RQ2)}

Table~\ref{tab:dual} provides a direct view of how the two inference paths divide labor. 
At the default threshold $\delta=0.8$, the Memorizer resolves 68.0\% of queries with 77.7 conditional ACC, 
so most inputs can terminate on the fast path without triggering full LLM reasoning. 
At the same time, the large gap between the full model and \textit{Mem.-only} (68.2 vs.\ 38.2 ACC) 
shows that direct memory matching alone cannot absorb the compositional burden of multi-hop QA. 
The comparison with \textit{Cog.-only} is equally revealing: 
although it invokes the LLM on every query, it still trails the full system by 8.9 points 
and incurs nearly $5\times$ higher latency. This indicates that QA memory contributes 
more than routing. The top-$K$ retrieved items act as structured hyperedge pointers 
that narrow the evidence space before LLM inference, 
improving slow-path reasoning quality in addition to reducing cost. 
Figure~\ref{fig:acc_latency_tradeoff} further shows a smooth Pareto frontier as $\delta$ varies, 
with the default setting lying near the knee of the curve. 
The routing gate therefore functions as a controlled mechanism for balancing accuracy and latency, 
rather than a brittle heuristic.

\subsection{Ablation Study (RQ3)}

Figure~\ref{fig:fd_rag_ablation} shows that the gain of \textsc{FD-RAG} is distributed across representation, 
inference, and federation modules rather than dominated by a single design choice. 
Removing semantic-aware hypergraph learning (\textit{w/o SAHL}) 
consistently lowers accuracy and increases latency, indicating that the learned hypergraph is 
not an auxiliary construction step but the structural basis for precise memory formation and evidence 
localization. Removing QA memory (\textit{w/o QA memory}) is even more damaging on both axes, 
since QA memory is the interface shared by the Memorizer and the Cognizer: once it is removed, 
the model loses both efficient direct matching and focused grounding for the slow path. 
The two inference ablations reveal complementary failure modes. 
\textit{w/o Memorizer} mainly hurts efficiency by forcing all queries through LLM-based reasoning, 
whereas \textit{w/o Cognizer} causes the sharpest accuracy collapse, 
showing that memory matching alone cannot support compositional reasoning. Finally, the drop of \textit{w/o Fed-KnowAgg} 
confirms that global memory fusion remains important when supporting evidence is distributed across clients. 
Additional experimental results are provided in Appendix~\ref{additional-result}.

\section{Conclusion}
\label{sec:conclusion}

We presented FD-RAG, a federated dual-system RAG framework for edge environments. 
By organizing local corpora into semantic hypergraphs and distilling 
them into QA memories, FD-RAG enables efficient memory-based 
response for well-covered queries while reserving LLM-based reasoning 
for more complex cases. This fast--slow decoupling, 
together with federated memory aggregation, 
provides a practical way to improve knowledge utilization 
under data locality and resource constraints.
Empirical results validate the effectiveness of FD-RAG in balancing accuracy and 
efficiency, while the proposed hypergraph learning 
objective admits convergence guarantees. 
We hope FD-RAG offers a principled foundation for RAG deployment in real-world edge scenarios.

\section*{Limitations}
\label{sec:limitations}

FD-RAG demonstrates strong performance in edge-device application scenarios; however, its adaptation to entirely new domains or tasks still relies heavily on an offline construction and optimization process. This limitation arises because the system's knowledge base, memory organization, and reasoning patterns are built from existing data during offline preparation. As a result, when a new domain introduces unfamiliar concepts, terminologies, or relational structures, the system typically requires offline rebuilding, retraining, or re-indexing before it can achieve strong performance, rather than adapting immediately during deployment. While this design helps keep online inference efficient, it also limits the speed and flexibility of domain transfer. Future work will therefore focus on developing more incremental and adaptive mechanisms to reduce the cost of offline reconstruction. In particular, transfer learning and lightweight continual updating may help transfer accumulated knowledge from known domains to new ones more efficiently, thereby improving adaptation speed and system robustness.

\bibliography{custom}

@article{lewis2020retrieval,
  title={Retrieval-augmented generation for knowledge-intensive nlp tasks},
  author={Lewis, Patrick and Perez, Ethan and Piktus, Aleksandra and Petroni, Fabio and Karpukhin, Vladimir and Goyal, Naman and K{\"u}ttler, Heinrich and Lewis, Mike and Yih, Wen-tau and Rockt{\"a}schel, Tim and others},
  journal={Advances in neural information processing systems},
  volume={33},
  pages={9459--9474},
  year={2020}
}

@article{gao2023retrieval,
  title={Retrieval-augmented generation for large language models: A survey},
  author={Gao, Yunfan and Xiong, Yun and Gao, Xinyu and Jia, Kangxiang and Pan, Jinliu and Bi, Yuxi and Dai, Yixin and Sun, Jiawei and Wang, Haofen and Wang, Haofen},
  journal={arXiv preprint arXiv:2312.10997},
  volume={2},
  pages={1},
  year={2023}
}

@article{kahneman2003maps,
  title={Maps of bounded rationality: Psychology for behavioral economics},
  author={Kahneman, Daniel},
  journal={American economic review},
  volume={93},
  number={5},
  pages={1449--1475},
  year={2003},
  publisher={American Economic Association}
}

@article{luo2025hypergraphrag,
  title={Hypergraphrag: Retrieval-augmented generation via hypergraph-structured knowledge representation},
  author={Luo, Haoran and Chen, Guanting and Zheng, Yandan and Wu, Xiaobao and Guo, Yikai and Lin, Qika and Feng, Yu and Kuang, Zemin and Song, Meina and Zhu, Yifan and others},
  journal={Advances in Neural Information Processing Systems},
  volume={38},
  pages={152206--152234},
  year={2026}
}

@inproceedings{jin2024long,
  title={Long-context llms meet rag: Overcoming challenges for long inputs in rag},
  author={Jin, Bowen and Yoon, Jinsung and Han, Jiawei and Arik, Sercan},
  booktitle={International Conference on Learning Representations},
  volume={2025},
  pages={37784--37822},
  year={2025}
}

@article{li2024long,
  title={Long Context vs. RAG for LLMs: An Evaluation and Revisits},
  author={Li, Xinze and Cao, Yixin and Ma, Yubo and Sun, Aixin},
  journal={arXiv preprint arXiv:2501.01880},
  year={2024}
}

@inproceedings{zhuang2024efficientrag,
  title={Efficientrag: Efficient retriever for multi-hop question answering},
  author={Zhuang, Ziyuan and Zhang, Zhiyang and Cheng, Sitao and Yang, Fangkai and Liu, Jia and Huang, Shujian and Lin, Qingwei and Rajmohan, Saravan and Zhang, Dongmei and Zhang, Qi},
  booktitle={Proceedings of the 2024 Conference on Empirical Methods in Natural Language Processing},
  pages={3392--3411},
  year={2024}
}

@inproceedings{liu2025hoprag,
  title={Hoprag: Multi-hop reasoning for logic-aware retrieval-augmented generation},
  author={Liu, Hao and Wang, Zhengren and Chen, Xi and Li, Zhiyu and Xiong, Feiyu and Yu, Qinhan and Zhang, Wentao},
  booktitle={Findings of the Association for Computational Linguistics: ACL 2025},
  pages={1897--1913},
  year={2025}
}

@inproceedings{asai2023self,
  title={Self-rag: Learning to retrieve, generate, and critique through self-reflection},
  author={Asai, Akari and Wu, Zeqiu and Wang, Yizhong and Sil, Avi and Hajishirzi, Hannaneh},
  booktitle={International conference on learning representations},
  volume={2024},
  pages={9112--9141},
  year={2024}
}

@inproceedings{yueinference,
  title={Inference scaling for long-context retrieval augmented generation},
  author={Yue, Zhenrui and Zhuang, Honglei and Bai, Aijun and Hui, Kai and Jagerman, Rolf and Zeng, Hansi and Qin, Zhen and Wang, Dong and Wang, Xuanhui and Bendersky, Michael},
  booktitle={International Conference on Learning Representations},
  volume={2025},
  pages={72914--72938},
  year={2025}
}

@inproceedings{yang2018hotpotqa,
  title={HotpotQA: A dataset for diverse, explainable multi-hop question answering},
  author={Yang, Zhilin and Qi, Peng and Zhang, Saizheng and Bengio, Yoshua and Cohen, William and Salakhutdinov, Ruslan and Manning, Christopher D},
  booktitle={Proceedings of the 2018 conference on empirical methods in natural language processing},
  pages={2369--2380},
  year={2018}
}

@inproceedings{ho2020constructing,
  title={Constructing a multi-hop qa dataset for comprehensive evaluation of reasoning steps},
  author={Ho, Xanh and Nguyen, Anh-Khoa Duong and Sugawara, Saku and Aizawa, Akiko},
  booktitle={Proceedings of the 28th International Conference on Computational Linguistics},
  pages={6609--6625},
  year={2020}
}

@article{trivedi2022musique,
  title={MuSiQue: Multi-hop Questions via Single-hop Question Composition},
  author={Trivedi, Harsh and Balasubramanian, Niranjan and Khot, Tushar and Sabharwal, Ashish},
  journal={Transactions of the Association for Computational Linguistics},
  volume={10},
  pages={539--554},
  year={2022}
}

@inproceedings{bai2024longbench,
  title={Longbench: A bilingual, multitask benchmark for long context understanding},
  author={Bai, Yushi and Lv, Xin and Zhang, Jiajie and Lyu, Hongchang and Tang, Jiankai and Huang, Zhidian and Du, Zhengxiao and Liu, Xiao and Zeng, Aohan and Hou, Lei and others},
  booktitle={Proceedings of the 62nd annual meeting of the association for computational linguistics (volume 1: Long papers)},
  pages={3119--3137},
  year={2024}
}

@article{jiang2024longrag,
  title={Longrag: Enhancing retrieval-augmented generation with long-context llms},
  author={Jiang, Ziyan and Ma, Xueguang and Chen, Wenhu},
  journal={arXiv preprint arXiv:2406.15319},
  year={2024}
}

@inproceedings{sarthi2024raptor,
  title={Raptor: Recursive abstractive processing for tree-organized retrieval},
  author={Sarthi, Parth and Abdullah, Salman and Tuli, Aditi and Khanna, Shubh and Goldie, Anna and Manning, Christopher},
  booktitle={International Conference on Learning Representations},
  volume={2024},
  pages={32628--32649},
  year={2024}
}

@article{grattafiori2024llama,
  title={The llama 3 herd of models},
  author={Grattafiori, Aaron and Dubey, Abhimanyu and Jauhri, Abhinav and Pandey, Abhinav and Kadian, Abhishek and Al-Dahle, Ahmad and Letman, Aiesha and Mathur, Akhil and Schelten, Alan and Vaughan, Alex and others},
  journal={arXiv preprint arXiv:2407.21783},
  year={2024}
}

@article{chen2024bge,
  title={BGE M3-Embedding: Multi-Lingual, Multi-Functionality, Multi-Granularity Text Embeddings Through Self-Knowledge Distillation},
  author={Chen, Jianlv and Xiao, Shitao and Zhang, Peitian and Luo, Kun and Lian, Defu and Liu, Zheng},
  journal={arXiv e-prints},
  pages={arXiv--2402},
  year={2024}
}

@inproceedings{li2024graphreader,
  title={GraphReader: Building Graph-based Agent to Enhance Long-Context Abilities of Large Language Models},
  author={Li, Shilong and He, Yancheng and Guo, Hangyu and Bu, Xingyuan and Bai, Ge and Liu, Jie and Liu, Jiaheng and Qu, Xingwei and Li, Yangguang and Ouyang, Wanli},
  booktitle={Findings of the Association for Computational Linguistics: EMNLP 2024},
  pages={12758--12786},
  year={2024}
}

@article{guinet2024automated,
  title={Automated evaluation of retrieval-augmented language models with task-specific exam generation},
  author={Guinet, Gauthier and Omidvar-Tehrani, Behrooz and Deoras, Anoop and Callot, Laurent},
  journal={arXiv preprint arXiv:2405.13622},
  year={2024}
}

@article{yang2024crag,
  title={CRAG--Comprehensive RAG Benchmark},
  author={Yang, Xiao and Sun, Kai and Xin, Hao and Sun, Yushi and Bhalla, Nikita and Chen, Xiangsen and Choudhary, Sajal and Gui, Rongze Daniel and Jiang, Ziran Will and Jiang, Ziyu},
  journal={arXiv preprint arXiv:2406.04744},
  year={2024}
}

@article{shang2024ada,
  title={Ada-MSHyper: adaptive multi-scale hypergraph transformer for time series forecasting},
  author={Shang, Zongjiang and Chen, Ling and Wu, Binqing and Cui, Dongliang},
  journal={Advances in Neural Information Processing Systems},
  volume={37},
  pages={33310--33337},
  year={2024}
}

@article{hochman1969pareto,
  title={Pareto optimal redistribution},
  author={Hochman, Harold M and Rodgers, James D},
  journal={The American economic review},
  volume={59},
  number={4},
  pages={542--557},
  year={1969},
  publisher={JSTOR}
}

@article{edge2024local,
  title={From local to global: A graph rag approach to query-focused summarization},
  author={Edge, Darren and Trinh, Ha and Cheng, Newman and Bradley, Joshua and Chao, Alex and Mody, Apurva and Truitt, Steven and Metropolitansky, Dasha and Ness, Robert Osazuwa and Larson, Jonathan},
  journal={arXiv preprint arXiv:2404.16130},
  year={2024}
}

@article{liu2025argus,
  title={Argus: Federated Non-convex Bilevel Learning over 6G Space-Air-Ground Integrated Network},
  author={Liu, Ya and Yang, Kai and Zhu, Yu and Yang, Keying and Zhao, Haibo},
  journal={arXiv preprint arXiv:2505.09106},
  year={2025}
}

@article{ma2025cellular,
  title={Cellular Traffic Prediction via Byzantine-robust Asynchronous Federated Learning},
  author={Ma, Hui and Yang, Kai and Jiao, Yang},
  journal={IEEE Transactions on Network Science and Engineering},
  year={2025},
  publisher={IEEE}
}

@inproceedings{liu2025nested,
  title={A nested zeroth-order fine-tuning approach for cloud-edge llm agents},
  author={Liu, Ya and Yang, Kai and Zhu, Yu and Yang, Keying and Jian, Chengtao and Ni, Wuguang and Ye, Xiaozhou and Ouyang, Ye},
  booktitle={Pacific-Asia Conference on Knowledge Discovery and Data Mining},
  pages={3--15},
  year={2025},
  organization={Springer}
}

@inproceedings{jiao2025pr,
  title={Pr-attack: Coordinated prompt-rag attacks on retrieval-augmented generation in large language models via bilevel optimization},
  author={Jiao, Yang and Wang, Xiaodong and Yang, Kai},
  booktitle={Proceedings of the 48th International ACM SIGIR Conference on Research and Development in Information Retrieval},
  pages={656--667},
  year={2025}
}

@inproceedings{wang2024feb4rag,
  title={Feb4rag: Evaluating federated search in the context of retrieval augmented generation},
  author={Wang, Shuai and Khramtsova, Ekaterina and Zhuang, Shengyao and Zuccon, Guido},
  booktitle={Proceedings of the 47th International ACM SIGIR Conference on Research and Development in Information Retrieval},
  pages={763--773},
  year={2024}
}

@article{liang2026fedmosaic,
  title={FedMosaic: Federated Retrieval-Augmented Generation via Parametric Adapters},
  author={Liang, Zhilin and Wang, Yuxiang and Zhou, Zimu and Zhang, Hainan and Liu, Boyi and Tong, Yongxin},
  journal={arXiv preprint arXiv:2602.05235},
  year={2026}
}

@inproceedings{he2025pfedrag,
  title={pFedRAG: A Personalized Federated Retrieval-Augmented Generation System with Depth-Adaptive Tiered Embedding Tuning},
  author={He, Hangyu and Yuan, Xin and Wu, Kai and Liu, Ren Ping and Ni, Wei},
  booktitle={Findings of the Association for Computational Linguistics: EMNLP 2025},
  year={2025}
}

@article{fajardo2025fedragframework,
  title={FedRAG: A Framework for Fine-Tuning Retrieval-Augmented Generation Systems},
  author={Fajardo, Val Andrei and Emerson, David B. and Singh, Amandeep and Chatrath, Veronica and Lotif, Marcelo and Desetty, Ravi Theja and Cheung, Chi Ho and Matsuba, Izuki},
  journal={arXiv preprint arXiv:2506.09200},
  year={2025}
}

@inproceedings{guerraoui2025ragroute,
  title={Efficient federated search for retrieval-augmented generation},
  author={Guerraoui, Rachid and Kermarrec, Anne-Marie and Petrescu, Diana and Pires, Rafael and Randl, Mathis and de Vos, Martijn},
  booktitle={Proceedings of the 5th Workshop on Machine Learning and Systems},
  pages={74--81},
  year={2025}
}

@article{xu2024cfedrag,
  title={C-FedRAG: A Confidential Federated Retrieval-Augmented Generation System},
  author={Xu, Ziyue},
  journal={arXiv preprint arXiv:2412.13163},
  year={2024}
}

@article{zhao2024frag,
  title={Frag: Toward federated vector database management for collaborative and secure retrieval-augmented generation},
  author={Zhao, Dongfang},
  journal={arXiv preprint arXiv:2410.13272},
  year={2024}
}

@inproceedings{gutierrez2025from,
  title={From RAG to Memory: Non-Parametric Continual Learning for Large Language Models},
  author={Guti{\'e}rrez, Bernal Jim{\'e}nez and Shu, Yiheng and Qi, Weijian and Zhou, Sizhe and Su, Yu},
  booktitle={International Conference on Machine Learning},
  pages={21497--21515},
  year={2025},
  organization={PMLR}
}

@misc{nvidia_jetson_orin,
  author = {{NVIDIA Corporation}},
  title = {NVIDIA Jetson Orin},
  year = {2022},
  url = {https://www.nvidia.com/en-us/autonomous-machines/embedded-systems/jetson-orin/},
  note = {Accessed: 2026-03-24}
}

@article{honnibal2020spacy,
  title={spaCy: Industrial-strength natural language processing in Python},
  author={Honnibal, Matthew and Montani, Ines and Van Landeghem, Sofie and Boyd, Adriane and others},
  year={2020},
  publisher={Zenodo, Honolulu, HI, USA}
}

@article{fan2025minirag,
  title={Minirag: Towards extremely simple retrieval-augmented generation},
  author={Fan, Tianyu and Wang, Jingyuan and Ren, Xubin and Huang, Chao},
  journal={arXiv preprint arXiv:2501.06713},
  year={2025}
}

@inproceedings{liu2024mobilellm,
title={Mobile{LLM}: Optimizing Sub-billion Parameter Language Models for On-Device Use Cases},
author={Zechun Liu and Changsheng Zhao and Forrest Iandola and Chen Lai and Yuandong Tian and Igor Fedorov and Yunyang Xiong and Ernie Chang and Yangyang Shi and Raghuraman Krishnamoorthi and Liangzhen Lai and Vikas Chandra},
booktitle={Forty-first International Conference on Machine Learning},
year={2024}
}

@article{xu2025distributed,
  title={Distributed retrieval-augmented generation},
  author={Xu, Chenhao and Gao, Longxiang and Miao, Yuan and Zheng, Xi},
  journal={arXiv preprint arXiv:2505.00443},
  year={2025}
}

@inproceedings{chakraborty-etal-2025-federated,
    title = "Federated Retrieval-Augmented Generation: A Systematic Mapping Study",
    author = "Chakraborty, Abhijit  and
      Dahal, Chahana  and
      Gupta, Vivek",
    booktitle = "Findings of the Association for Computational Linguistics: EMNLP 2025",
    month = nov,
    year = "2025",
    publisher = "Association for Computational Linguistics",
    pages = "7362--7374"
}

@article{oche2025systematic,
  title={A systematic review of key retrieval-augmented generation (rag) systems: Progress, gaps, and future directions},
  author={Oche, Agada Joseph and Folashade, Ademola Glory and Ghosal, Tirthankar and Biswas, Arpan},
  journal={arXiv preprint arXiv:2507.18910},
  year={2025}
}

@book{boyd2004convex,
  title={Convex optimization},
  author={Boyd, Stephen and Vandenberghe, Lieven},
  year={2004},
  publisher={Cambridge university press}
}

@article{salton1989automatic,
  title={Automatic text processing: The transformation, analysis, and retrieval of},
  author={Salton, Gerard},
  journal={Reading: Addison-Wesley},
  volume={169},
  year={1989}
}

@article{warner1965randomized,
  title={Randomized response: A survey technique for eliminating evasive answer bias},
  author={Warner, Stanley L},
  journal={Journal of the American statistical association},
  volume={60},
  number={309},
  pages={63--69},
  year={1965},
  publisher={Taylor \& Francis}
}

@inproceedings{dwork2008differential,
  title={Differential privacy: A survey of results},
  author={Dwork, Cynthia},
  booktitle={International conference on theory and applications of models of computation},
  pages={1--19},
  year={2008},
  organization={Springer}
}

\appendix

% !TEX root = ../acl_lualatex.tex
\section{Proof of Proposition~\ref{prop:main-convergence}}
\label{appendix:convergence-proof}

\subsection{Problem Setup}

We analyze the convergence of the parameter \( H^t \in \mathbb{R}^n \) in an hypergraph learning framework. The variable \( H^t \) is constrained to lie on the probability simplex \( \Delta^n \), defined as follows.

\begin{definition}[Probability Simplex]
The $n$-dimensional probability simplex is
\(
\Delta^n := \{ H \in \mathbb{R}^n \mid H_i \geq 0 \text{ for all } i,\ \sum_{i=1}^n H_i = 1 \}.
\)
\end{definition}

We consider the optimization problem \( \min_{H^t \in \Delta^n} L_{\text{total}}(H^t) \), where \( L_{\text{total}} : \mathbb{R}^n \to \mathbb{R} \) is a differentiable objective function.

\begin{assumption}[Smoothness]
The objective function \( L_{\text{total}} \) is $L$-smooth, i.e., \( \| \nabla L_{\text{total}}(x) - \nabla L_{\text{total}}(y) \| \leq L \| x - y \| \) for all \( x, y \in \mathbb{R}^n \).
\end{assumption}

We employ Projected Gradient Descent (PGD) to solve the above problem.

\begin{definition}[Projected Gradient Descent (PGD)]
Given a step size \( \eta > 0 \), the PGD update rule is
\begin{equation}
    H^t_{k+1} = \operatorname{Proj}_{\Delta^n} \left( H^t_k - \eta \nabla L_{\text{total}}(H^t_k) \right),
\end{equation}
where \( \operatorname{Proj}_{\Delta^n}(\cdot) \) denotes the Euclidean projection onto \( \Delta^n \).
\end{definition}

\begin{definition}[Gradient Mapping]
The gradient mapping at iteration \( k \) is \( G^t_k := \frac{H^t_k - H^t_{k+1}}{\eta} \).
\end{definition}

\subsection{Descent Property and Convergence}

We now establish the descent property of the PGD updates.

\begin{lemma}[Descent Lemma for PGD]
\label{lemma:descent}
Let \( \eta = \frac{1}{L} \). Then, for each iteration \( k \), the following holds:
\begin{equation}
    L_{\text{total}}(H^t_{k+1}) \leq L_{\text{total}}(H^t_k) - \frac{1}{2L} \| G^t_k \|^2.
\end{equation}
\end{lemma}

\begin{proof}
For brevity, let \( L_k := L_{\text{total}}(H^t_k) \) and \( g_k := \nabla L_{\text{total}}(H^t_k) \). By $L$-smoothness,
\begin{equation}
    L_{k+1} \leq L_k + g_k^\top (H^t_{k+1} - H^t_k) + \frac{L}{2} \| H^t_{k+1} - H^t_k \|^2.
\end{equation}

Substituting \( H^t_{k+1} - H^t_k = -\eta G^t_k \), we obtain
\begin{equation}
    L_{k+1} \leq L_k - \eta g_k^\top G^t_k + \frac{L \eta^2}{2} \| G^t_k \|^2.
\end{equation}

Next, the optimality condition of Euclidean projection gives
\begin{equation}
    \left\langle H^t_k - \eta g_k - H^t_{k+1},\, y - H^t_{k+1} \right\rangle \leq 0, \quad \forall y \in \Delta^n.
\end{equation}
Setting \( y = H^t_k \) and using the definition of \( G^t_k \), we get
\begin{equation}
    g_k^\top G^t_k \geq \| G^t_k \|^2.
\end{equation}
Substituting this bound back yields
\begin{equation}
    L_{k+1} \leq L_k - \eta \| G^t_k \|^2 + \frac{L \eta^2}{2} \| G^t_k \|^2.
\end{equation}
Choosing \( \eta = \frac{1}{L} \), we obtain \( L_{k+1} \leq L_k - \frac{1}{2L} \| G^t_k \|^2 \), which is exactly the desired inequality.
\end{proof}

\paragraph{Restatement of Proposition~\ref{prop:main-convergence}.}
Let \( H^t_* \in \Delta^n \) be the optimal solution and define \( \Delta := L_{\text{total}}(H^t_0) - L_{\text{total}}(H^t_*) \). Then, after \( T \) iterations of PGD with \( \eta = \frac{1}{L} \), the minimum norm of the gradient mapping satisfies:
\begin{equation}
    \min_{0 \leq k < T} \| G^t_k \|^2 \leq \frac{2L \Delta}{T}.
\end{equation}
Consequently, to achieve \( \| G^t_k \| \leq \epsilon \), the number of iterations required is:
\begin{equation}
    T \geq \frac{2L \Delta}{\epsilon^2},
\end{equation}
i.e., the iteration complexity is \( \mathcal{O}(1/\epsilon^2) \).

\begin{proof}
Summing the descent inequality from Lemma~\ref{lemma:descent} over \( k = 0 \) to \( T - 1 \), we obtain
\(
\sum_{k=0}^{T-1} \frac{1}{2L} \| G^t_k \|^2 \leq L_{\text{total}}(H^t_0) - L_{\text{total}}(H^t_T) \leq \Delta
\),
and hence \( \sum_{k=0}^{T-1} \| G^t_k \|^2 \leq 2L \Delta \). Therefore,
\(
\min_{0 \leq k < T} \| G^t_k \|^2
\leq \frac{1}{T} \sum_{k=0}^{T-1} \| G^t_k \|^2
\leq \frac{2L \Delta}{T}
\).
Solving for \( T \) such that the right-hand side is at most \( \epsilon^2 \) gives the desired result.
\end{proof}
\section{Proof of Proposition~\ref{prop:main-ldp}}
\label{appendix:ldp-proof}
\subsection{Preliminaries}

\begin{definition}[Local Differential Privacy]
A randomized mechanism $\mathcal{M} \colon \mathcal{X} \to \mathcal{Y}$ satisfies $\epsilon$-local differential privacy ($\epsilon$-LDP) if, for all inputs $x, x' \in \mathcal{X}$ and all outputs $y \in \mathcal{Y}$, it holds that
\[
\frac{\Pr[\mathcal{M}(x) = y]}{\Pr[\mathcal{M}(x') = y]} \leq e^{\epsilon}.
\]
\end{definition}

This condition ensures that the mechanism's output does not reveal significant information about any individual input, thereby preserving privacy in the local model.

\subsection{Restatement of Proposition~\ref{prop:main-ldp}}

\paragraph{Proposition~\ref{prop:main-ldp}.}
Let $e$ be a sensitive entity and $W$ a candidate set with $|W| = c$. Consider the perturbation mechanism defined by the conditional distribution
\[
\Pr[e' \mid e] =
\begin{cases}
\frac{e^{\epsilon}}{e^{\epsilon} + c - 1}, & \text{if } e' = e, \\
\frac{1}{e^{\epsilon} + c - 1}, & \text{if } e' \in W, \, e' \neq e.
\end{cases}
\]
Then, this mechanism satisfies $\epsilon$-local differential privacy.

\subsection{Proof of Proposition~\ref{prop:main-ldp}}

We prove the proposition by verifying the $\epsilon$-LDP condition for all possible inputs $e, e^* \in \{e\} \cup W$ and outputs $e' \in \{e\} \cup W$. Specifically, we must show:
\[
\frac{\Pr[e' \mid e]}{\Pr[e' \mid e^*]} \leq e^{\epsilon}.
\]

\begin{proof}
We consider the following exhaustive cases:

\begin{itemize}
    \item \textbf{Case 1: $e' = e = e^*$.} \\
    \[
    \frac{\Pr[e' \mid e]}{\Pr[e' \mid e^*]} = \frac{\frac{e^{\epsilon}}{e^{\epsilon} + c - 1}}{\frac{e^{\epsilon}}{e^{\epsilon} + c - 1}} = 1 \leq e^{\epsilon}.
    \]

    \item \textbf{Case 2: $e' = e \neq e^*$.} \\
    \[
    \frac{\Pr[e' \mid e]}{\Pr[e' \mid e^*]} = \frac{\frac{e^{\epsilon}}{e^{\epsilon} + c - 1}}{\frac{1}{e^{\epsilon} + c - 1}} = e^{\epsilon}.
    \]

    \item \textbf{Case 3: $e' = e^* \neq e$.} \\
    \[
    \frac{\Pr[e' \mid e]}{\Pr[e' \mid e^*]} = \frac{\frac{1}{e^{\epsilon} + c - 1}}{\frac{e^{\epsilon}}{e^{\epsilon} + c - 1}} = \frac{1}{e^{\epsilon}} \leq e^{\epsilon}.
    \]

    \item \textbf{Case 4: $e' \neq e$ and $e' \neq e^*$.} \\
    \[
    \frac{\Pr[e' \mid e]}{\Pr[e' \mid e^*]} = \frac{\frac{1}{e^{\epsilon} + c - 1}}{\frac{1}{e^{\epsilon} + c - 1}} = 1 \leq e^{\epsilon}.
    \]
\end{itemize}

In all cases, the privacy condition $\frac{\Pr[e' \mid e]}{\Pr[e' \mid e^*]} \leq e^{\epsilon}$ is satisfied. Therefore, the mechanism guarantees $\epsilon$-local differential privacy.
\end{proof}
\section{Experiments Details}
\label{appendix:experiment-details}
\subsection{Statistics of Datasets}
Table~\ref{tab:benchmark_stats} provides detailed statistics of the datasets used in our experiments, including HotPotQA, 2WikiMQA, and MuSiQue. These datasets vary in size and complexity, offering a comprehensive evaluation framework for multi-hop question answering models. 
\label{app:benchmark}
\begin{table*}[htbp]
    \centering
    \caption{Statistics of Datasets.}
    \label{tab:benchmark_stats}
        \begin{tabular}{lccc}
            \toprule
            \textbf{Dataset} & \textbf{Avg \#Tokens} & \textbf{Max \#Tokens} & \textbf{\#Samples} \\
            \midrule
            \textbf{HotPotQA}        & 9.1k  & 12.7k & 500 \\
            \textbf{2WikiMQA}        & 9.2k  & 12.3k & 500 \\
            \textbf{MuSiQue}         & 11.1k & 17.3k & 500 \\
            \bottomrule
        \end{tabular}
\end{table*}
\subsection{Implementation Details}
\label{app:implementation}
We use $\mu=0.5$, $K=5$, $\lambda=0.6$, and $\delta=0.8$ across all datasets and clients. For privacy-preserving memory sharing, we set the local differential privacy budget to $\epsilon=1.0$ with candidate set size $c=5$. We optimize the semantic-aware hypergraph learning objective for 300 steps using a learning rate of 0.05, and keep the same hyperparameter configuration throughout training. We use BGE-M3~\citep{chen2024bge} as the dense encoder and \texttt{Llama-3.1-8B}~\citep{grattafiori2024llama} in INT4 quantized form as the language model. To ensure a consistent federated protocol~\citep{liu2025argus}, all federated experiments use five clients.

For the multi-granularity hypergraph module, we set the numbers of candidate hyperedges as $M^p=\lceil N_p/4 \rceil$ and $M^s=\lceil N_s/4 \rceil$, where $N_p$ and $N_s$ denote the numbers of paragraph- and sentence-level nodes, respectively. Each soft incidence matrix $\hat{H}^t$ is initialized with positive random values and row-wise normalized onto the probability simplex, after which the paragraph- and sentence-level hypergraphs are jointly optimized under the unified objective in Eq.~(7), while maintaining separate incidence matrices and hyperedge prototypes for the two granularities.

We deploy FD-RAG on a real NVIDIA Jetson Orin Nano 8GB edge device. Table~\ref{tab:eval_hardware} summarizes the evaluation hardware configuration. We follow the official device specifications and report end-to-end latency measured on this platform with batch size 1 under a 15~W power limit~\citep{nvidia_jetson_orin}.

\begin{table}[t]
    \centering
    \small
    \caption{Evaluation Hardware.}
    \label{tab:eval_hardware}
    \setlength{\tabcolsep}{6pt}
    \renewcommand{\arraystretch}{1.05}
    \begin{tabular}{@{}l p{0.72\linewidth}@{}}
        \toprule
        \textbf{Component} & \textbf{Specification} \\
        \midrule
        CPU & Cortex-A78AE, 1.2\,GHz, 6-core \\
        GPU & Ampere, 1024 CUDA cores, 32 Tensor cores, 625\,MHz \\
        Power Limit & 15\,W \\
        DRAM & 8\,GiB LPDDR5-4250 \\
        Storage & 512\,GB SD Card, UHS \\
        \bottomrule
    \end{tabular}
\end{table}

\subsection{Baselines}
\label{appendix:baselines}
We compare FD-RAG with baselines in local and federated settings. Methods without a native federated design are evaluated only locally; federated comparisons are limited to methods with explicit decentralized protocols.

%%%%D1. Training Convergence Curves
\begin{figure*}[t]
    \centering
    \begin{subfigure}[b]{0.32\textwidth}
        \centering
        \includegraphics[width=\textwidth]{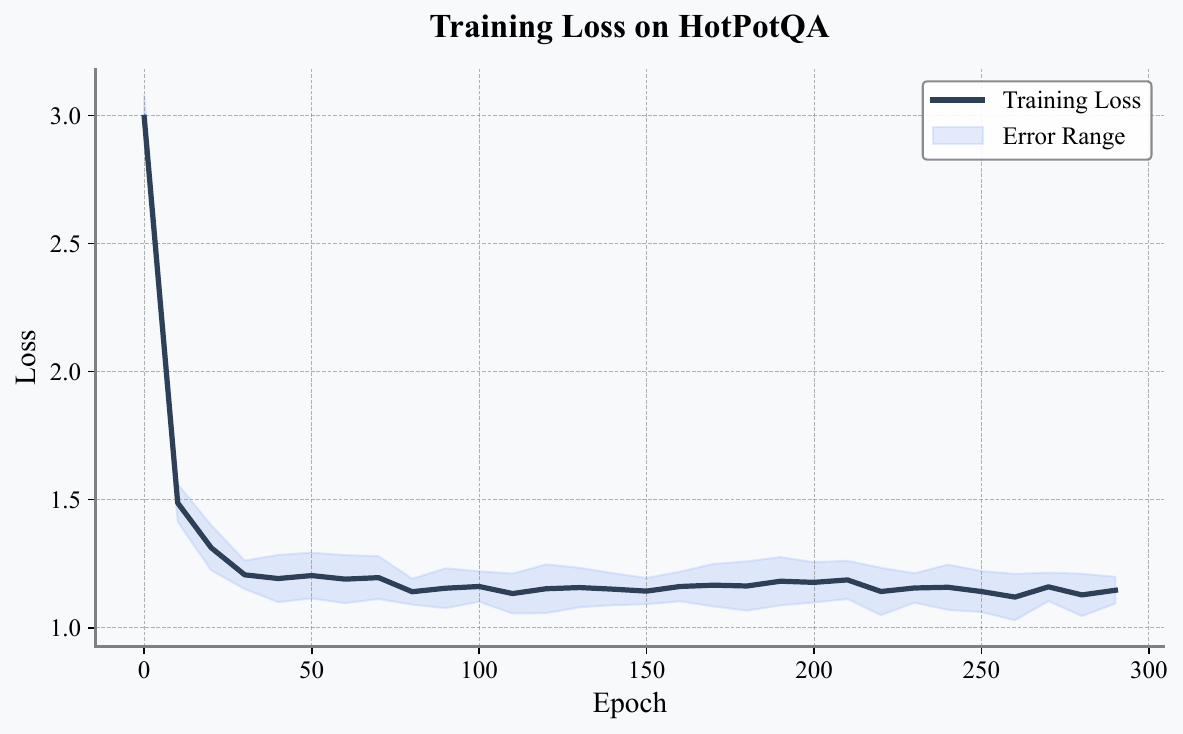}
        \caption{HotPotQA}
        \label{fig:hotpotqa}
    \end{subfigure}
    \hfill
    \begin{subfigure}[b]{0.32\textwidth}
        \centering
        \includegraphics[width=\textwidth]{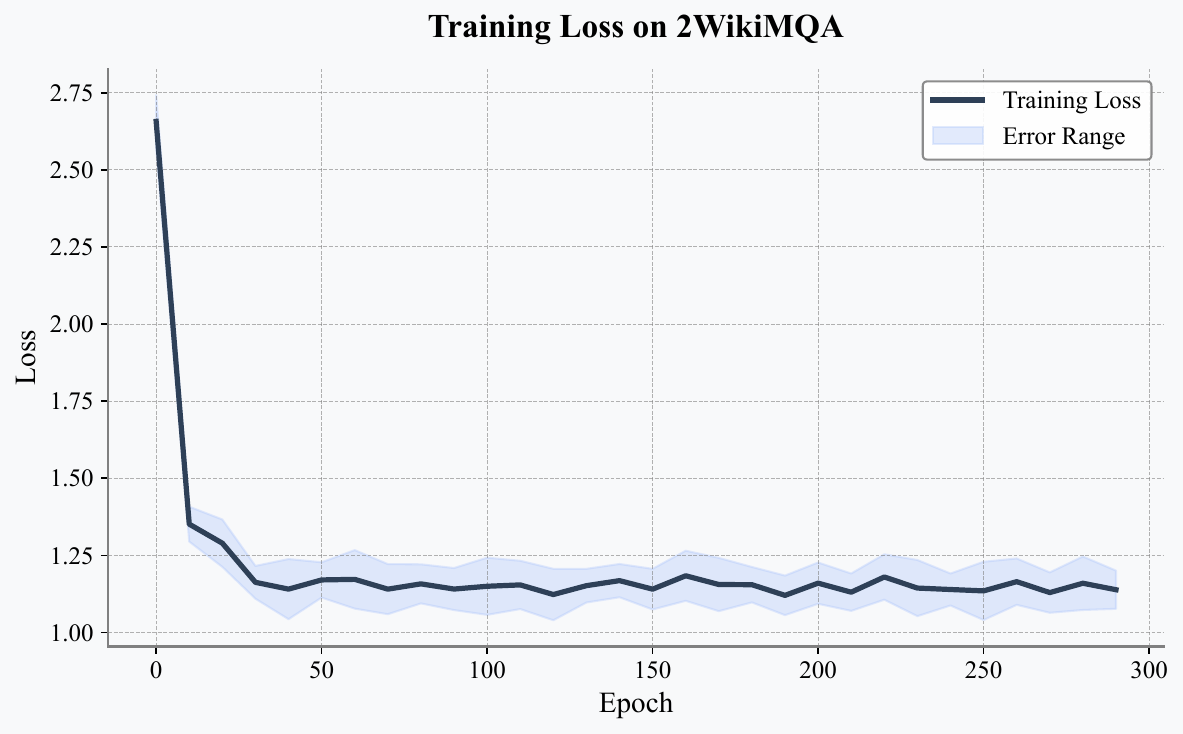}
        \caption{2WikiMQA}
        \label{fig:2wikimqa}
    \end{subfigure}
    \hfill
    \begin{subfigure}[b]{0.32\textwidth}
        \centering
        \includegraphics[width=\textwidth]{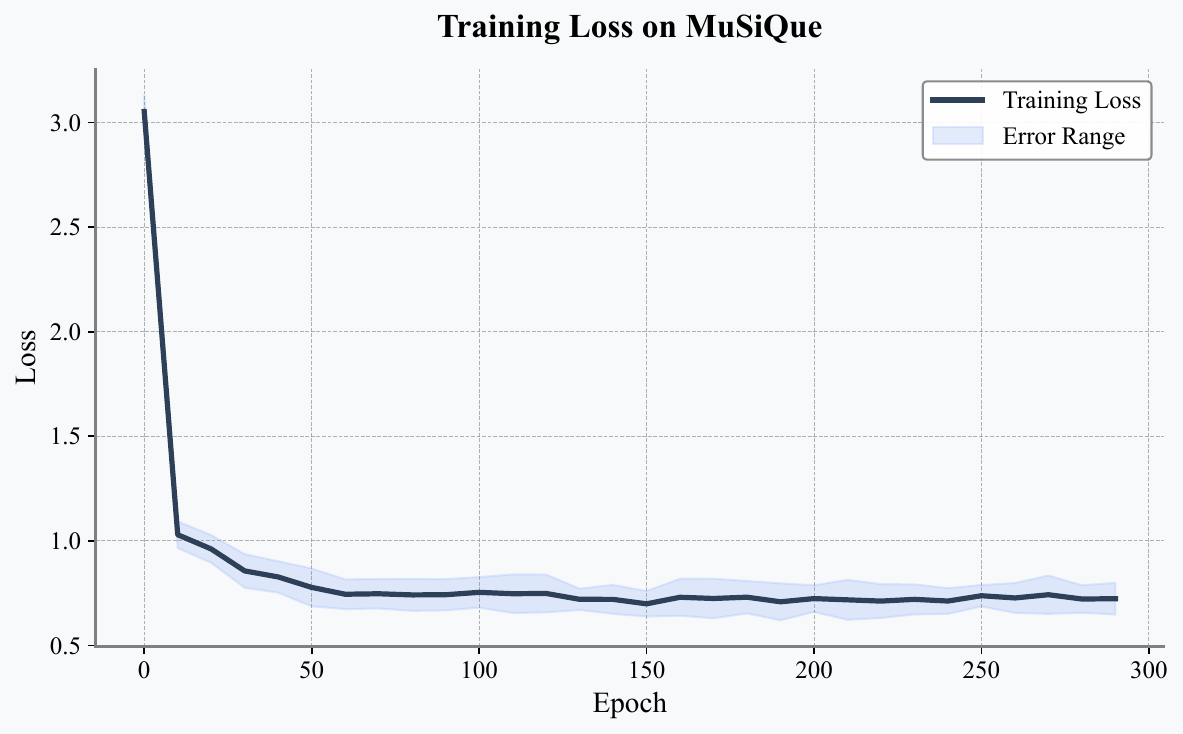}
        \caption{MuSiQue}
        \label{fig:musique}
    \end{subfigure}
    \caption{Training Loss Curves on Three QA Datasets: HotPotQA, 2WikiMQA, MuSiQue}
    \label{fig:datasets}
\end{figure*}

\begin{itemize}
    \item \textbf{Vanilla RAG} \citep{lewis2020retrieval}: Integrates a retriever with a generator, retrieving the top-5 relevant documents to augment context for improved answer generation.
    \item \textbf{IterDRAG} \citep{yueinference}: Segments complex queries into sub-queries, utilizing iterative retrieval and in-context learning to refine answers progressively through a reasoning chain.
    \item \textbf{LongRAG} \citep{jiang2024longrag}: Employs a dual-perspective approach, combining an information extractor, chain-of-thought-guided filter, and generator to address challenges in processing long texts and identifying fine-grained factual details.
    \item \textbf{EfficientRAG} \citep{zhuang2024efficientrag}: Trains lightweight models to iteratively generate queries and filter irrelevant information, enhancing retrieval efficiency and performance in multi-hop question answering.
    \item \textbf{RAPTOR} \citep{sarthi2024raptor}: Constructs a hierarchical summary tree through recursive embedding, clustering, and summarization, retrieving insights across abstraction levels to handle long documents effectively.
    \item \textbf{GraphRAG} \citep{edge2024local}: GraphRAG is a graph-based RAG method that constructs an entity knowledge graph and leverages community-level summaries to generate and aggregate responses for global queries.
    \item \textbf{HippoRAG 2} \citep{gutierrez2025from}: HippoRAG 2 builds a passage-linked knowledge graph and performs retrieval with Personalized PageRank, improving multi-hop evidence association while preserving strong factual recall.
    \item \textbf{HyperGraphRAG} \citep{luo2025hypergraphrag}: HyperGraphRAG is a hypergraph-based RAG method that models n-ary relations via hyperedges, enabling more accurate and efficient retrieval and generation than standard and graph-based RAG.
    \item \textbf{C-FedRAG} \citep{xu2024cfedrag}: C-FedRAG is a federated RAG framework that leverages confidential computing to enable secure and scalable retrieval across decentralized data sources.
    \item \textbf{RAGRoute} \citep{guerraoui2025ragroute}: Introduces a routing mechanism to dynamically select the most appropriate retrieval and generation strategies based on query characteristics.
\end{itemize}

%%%%%D3. Privacy Protection Evaluation
\begin{table*}[t]
\centering
\footnotesize
\caption{Privacy protection evaluation under an LLM restoration attack in the federated setting. FD-RAG Privacy uses the default budget $\epsilon=1.0$ with candidate set size $c=5$. Restoration Acc@1 measures how often the attacker correctly recovers the original sensitive entity from a shared QA/fact item; lower is better. QA ACC measures downstream utility after sharing; higher is better. Best results in each column are \textbf{bolded}; second best are \underline{underlined}.}
\label{tab:privacy}
\setlength{\tabcolsep}{4.2pt}
\renewcommand{\arraystretch}{1.12}
\begin{tabular*}{\textwidth}{@{\extracolsep{\fill}}lcccccccc@{}}
\toprule
\multirow{2}{*}{\textbf{Method}} &
\multicolumn{2}{c}{\textbf{HotPotQA}} &
\multicolumn{2}{c}{\textbf{2WikiMQA}} &
\multicolumn{2}{c}{\textbf{MuSiQue}} &
\multicolumn{2}{c}{\textbf{Average}} \\
\cmidrule(lr){2-3} \cmidrule(lr){4-5} \cmidrule(lr){6-7} \cmidrule(lr){8-9}
& \textbf{Rest.@1}$\downarrow$ & \textbf{ACC}$\uparrow$
& \textbf{Rest.@1}$\downarrow$ & \textbf{ACC}$\uparrow$
& \textbf{Rest.@1}$\downarrow$ & \textbf{ACC}$\uparrow$
& \textbf{Rest.@1}$\downarrow$ & \textbf{ACC}$\uparrow$ \\
\midrule
\textsc{No Protection}
& 95.1 & \textbf{65.4}
& 92.8 & \textbf{59.8}
& 89.4 & \textbf{36.5}
& 92.4 & \textbf{53.9} \\

\textsc{Type Masking}
& \textbf{3.9} & 59.0
& \textbf{3.4} & 53.7
& \textbf{2.8} & 31.9
& \textbf{3.4} & 48.2 \\

\textbf{FD-RAG Privacy (Ours)}
& \underline{9.6} & \underline{64.5}
& \underline{8.9} & \underline{58.9}
& \underline{7.8} & \underline{35.7}
& \underline{8.8} & \underline{53.0} \\
\bottomrule
\end{tabular*}
\end{table*}

\section{Additional Experimental Results}
\label{additional-result}
\subsection{Training Convergence Results}
\label{convergence}
We evaluated the convergence performance of our method across three distinct datasets. As illustrated in Figure~\ref{fig:datasets} , the loss curves of these datasets demonstrate the following trends: 1) Rapid Convergence: Across all datasets, the loss values sharply decline during the initial stages of training (within the first 50 epochs), indicating that the model swiftly captures thematic information and optimizes effectively early on. 2) Stability: As training progresses, the loss values stabilize after 100 epochs, suggesting that the model has largely converged without significant oscillations, which indicates the robustness of hypergraph learning across different datasets. 3) Consistency: The uniformity of the loss trends across the three datasets further supports the applicability of this method to various types of question-answering datasets, highlighting its adaptability and generalization capacity.

\subsection{Offline Construction Cost Analysis}
\label{app:offline_cost}
%%%% D2. Offline Construction Cost Comparison
\begin{table*}[t]
\centering
\small
\caption{Offline construction cost on HotPotQA under our deployment setting. All methods use \texttt{BGE-M3} and \texttt{Llama-3.1-8B} (INT4) on Jetson Orin Nano 8GB. OutTok/1kT denotes the number of LLM output tokens generated during construction per 1k corpus tokens. TP1kT denotes wall-clock construction time per 1k corpus tokens. Peak Mem. is measured during the offline construction stage. Lower is better for all metrics.}
\label{tab:offline_cost}
\setlength{\tabcolsep}{7pt}
\renewcommand{\arraystretch}{1.08}
\begin{tabular}{lccc}
\toprule
\textbf{Method} & \textbf{OutTok/1kT}$\downarrow$ & \textbf{TP1kT (s)}$\downarrow$ & \textbf{Peak Mem. (GB)}$\downarrow$ \\
\midrule
RAPTOR~\citep{sarthi2024raptor} & 88 & 12.7 & 6.0 \\
GraphRAG~\citep{edge2024local} & 542 & 39.4 & 7.4 \\
HippoRAG~2~\citep{gutierrez2025from} & 176 & 16.9 & 7.6 \\
HyperGraphRAG~\citep{luo2025hypergraphrag} & 214 & 19.3 & 7.2 \\
\textbf{FD-RAG (ours)} & \textbf{58} & \textbf{7.1} & \textbf{5.5} \\
\bottomrule
\end{tabular}
\end{table*}

We evaluate the offline construction cost of FD-RAG on HotPotQA and compare it with RAPTOR, GraphRAG, HippoRAG~2, and HyperGraphRAG under the same deployment setting as the main experiments. All methods use the same dense encoder (\texttt{BGE-M3}) and local generator (\texttt{Llama-3.1-8B}, INT4), and all measurements are collected on Jetson Orin Nano 8GB.

We report three metrics: LLM output tokens per 1k corpus tokens (\textbf{OutTok/1kT}), wall-clock construction time per 1k corpus tokens (\textbf{TP1kT}), and peak memory during construction. For FD-RAG, only QA-memory synthesis contributes to LLM token usage; \textsc{spaCy}-based fact extraction and hypergraph optimization are counted in runtime but not in token usage. Lower is better for all metrics.

As shown in Table~\ref{tab:offline_cost}, FD-RAG achieves the lowest cost on all three metrics, with 58 OutTok/1kT, 7.1 seconds TP1kT, and 5.5~GB peak memory. Compared with HippoRAG~2, FD-RAG reduces token usage, construction time, and peak memory by 67.0\%, 58.0\%, and 27.6\%, respectively. Compared with HyperGraphRAG, it reduces construction time from 19.3 to 7.1 seconds and peak memory from 7.2 to 5.5~GB. These results indicate that FD-RAG offers a more efficient offline construction profile and is better suited to resource-constrained edge deployment.

\subsection{Privacy Protection Evaluation}
\label{app:privacy_eval}

Given the growing security concerns around RAG systems~\citep{jiao2025pr}, we further assess whether memories sanitized before federation still leak recoverable sensitive entities. To this end, we extract person, organization, and location mentions from shared QA pairs and supporting-fact tuples, and perform an LLM restoration attack with \texttt{gpt-4o-mini}: given a sanitized item, the model is prompted to recover the original entity from context. We compare three sharing strategies: \textsc{No Protection}, which shares raw memories; \textsc{Type Masking}, which replaces each sensitive entity with a coarse placeholder such as \texttt{[PERSON]}; and \textsc{FD-RAG Privacy} (ours), which applies the proposed semantic candidate perturbation under $\epsilon$-LDP. We use Restoration Acc@1 to measure privacy leakage and downstream QA accuracy to measure utility. The $\epsilon=1.0$ setting for \textsc{FD-RAG Privacy} is the default federated configuration.

Table~\ref{tab:privacy} reveals a clear privacy--utility trade-off. \textsc{No Protection} yields restoration accuracy above 89 on all three benchmarks, indicating that raw shared memories leak entity identity almost directly. \textsc{Type Masking} minimizes leakage, reducing average Restoration Acc@1 to 3.4, but substantially degrades utility because coarse placeholders remove task-relevant semantics for retrieval and answer selection. In contrast, \textsc{FD-RAG Privacy} keeps the attacker success rate below 10\% on every dataset while preserving nearly the same QA performance as raw sharing. Averaged across datasets, it reduces Restoration Acc@1 from 92.4 to 8.8, a 90.5\% relative reduction in leakage compared with \textsc{No Protection}, while retaining 98.3\% of its utility (53.0 vs.\ 53.9 average ACC). Relative to \textsc{Type Masking}, it improves downstream accuracy by 4.8 points on average at the cost of only a modest increase in restoration accuracy. These results show that FD-RAG Privacy offers a substantially better privacy--utility balance than either raw sharing or coarse masking.

\begin{table}[t]
\centering
\footnotesize
\caption{Privacy--utility trade-off of FD-RAG Privacy under different local privacy budgets $\epsilon$ in the federated setting. We fix the candidate set size to $c=5$ and report averages over HotPotQA, 2WikiMQA, and MuSiQue. The default setting used in the main experiments is $\epsilon=1.0$. Lower Restoration Acc@1 indicates less leakage; higher QA ACC indicates better utility.}
\label{tab:privacy_epsilon}
\setlength{\tabcolsep}{8pt}
\renewcommand{\arraystretch}{1.1}
\begin{tabular}{lcc}
\toprule
\textbf{Privacy Budget} & \textbf{Rest.@1}$\downarrow$ & \textbf{ACC}$\uparrow$ \\
\midrule
$\epsilon=0.1$ & 4.7 & 49.4 \\
$\epsilon=0.5$ & 6.4 & 51.4 \\
\textbf{$\epsilon=1.0$ (default)} & 8.8 & 53.0 \\
$\epsilon=2.0$ & 14.3 & 53.6 \\
\bottomrule
\end{tabular}
\end{table}

Table~\ref{tab:privacy_epsilon} further examines this trade-off under different privacy budgets. As expected, larger $\epsilon$ weakens perturbation: the average restoration success rate rises from 4.7 at $\epsilon=0.1$ to 14.3 at $\epsilon=2.0$, while QA accuracy increases from 49.4 to 53.6. Notably, even under the strictest setting, FD-RAG Privacy still outperforms \textsc{Type Masking} in utility (49.4 vs.\ 48.2 average ACC), suggesting that semantically constrained perturbation preserves more task-relevant information than coarse placeholders. The default setting $\epsilon=1.0$ provides a balanced operating point: it keeps leakage an order of magnitude below raw sharing (8.8 vs.\ 92.4) while retaining nearly all downstream utility (53.0 vs.\ 53.9). Overall, FD-RAG supports a smooth and practical privacy--utility trade-off.

\subsection{Analysis of Generated QA Memory Questions}
Following prior work~\citep{guinet2024automated}, we conduct a statistical analysis of the generated QA-memory questions to assess their diversity and coverage. Concretely, for each dataset, we randomly sampled 100 source documents (contexts) and analyzed all generated questions derived from them along two dimensions:
%%%% D4. Generated Question Analysis
\begin{table*}[t]
\centering
\small
\caption{Number of Questions by Type Across Different Datasets}
\label{tab:question_types}
\begin{tabular}{l *{3}{S[table-format=4.0]}}
\toprule
\textbf{Question Type} & \textbf{HotPotQA} & \textbf{MuSiQue} & \textbf{2WikiMQA} \\
\midrule
Aggregation Questions           & 953  & 894  & 1188 \\
Comparison Questions            & 910  & 844  & 1105 \\
False Premise Questions         & 912  & 854  & 1097 \\
Multi-hop Questions             & 909  & 862  & 1143 \\
Post-processing Heavy Questions & 930  & 876  & 1152 \\
Set Questions                   & 1469 & 1459 & 1562 \\
\bottomrule
\end{tabular}
\end{table*}

\begin{figure*}[t]
    \centering
    \includegraphics[width=0.8\textwidth]{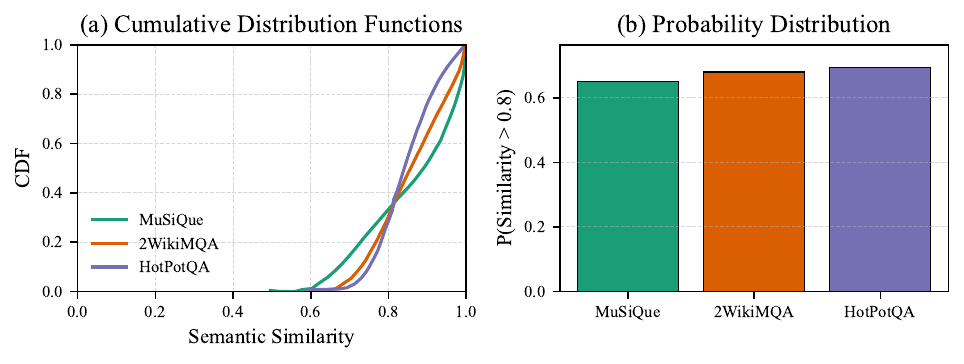}
    \caption{Performance evaluation of the generated questions with respect to semantic similarity.}
    \label{fig:cdf}
\end{figure*}

\paragraph{Question Type Analysis.} Drawing on several prior research efforts \citep{guinet2024automated,yang2024crag}, we further divided the complex questions generated by our QA memory construction pipeline into six categories: Set, Comparison, Aggregation, Multi-hop, Post-processing Heavy, and False Premise. This classification allows us to better assess the diversity and depth of QA memory question generation. As shown in Table~\ref{tab:question_types}, our approach consistently covers a wide range of question types across different benchmark datasets (HotPotQA, MuSiQue, and 2WikiMQA), demonstrating its strong generalization capability. In particular, the substantial presence of complex forms such as Set and Multi-hop questions suggests that the generated questions are not only varied in structure but also require multi-faceted reasoning and synthesis across information sources—an essential characteristic for evaluating evidence-grounded multi-hop QA.

\paragraph{Question Diversity and Coverage Analysis.} As shown in Figure~\ref{fig:cdf} (a), we present the cumulative distribution function (CDF) of semantic similarity between the original questions and the top-1 most semantically similar questions generated by our QA memory question generation method across three datasets: MuSiQue, 2WikiMQA, and HotPotQA. This distribution reflects how closely the generated questions align with the original questions at a semantic level. A wide and smooth distribution indicates that our method is capable of generating questions with varying degrees of similarity---ranging from highly similar to more diverse ones---demonstrating both strong semantic coverage and question diversity.

To quantify this, Figure~\ref{fig:cdf} (b) shows the proportion of generated questions whose semantic similarity with the original question exceeds 0.8. In all three datasets, this proportion exceeds 60\%, indicating that a substantial number of generated questions are semantically close to the originals. This high proportion suggests that our method effectively captures the core semantics of the input while also generating a broad range of diverse questions. Together, these results validate that our approach achieves a good balance between semantic fidelity and diversity, resulting in high-quality and broadly representative question generation.

\section{Prompt Templates}
\label{appendix:prompt} 
We provide the prompt templates used in our experiments. The prompts are designed to elicit specific information from the model, guiding it to generate accurate and relevant responses.

\begin{table*}
\begin{tcolorbox}[colback=gray!2!white, colframe=gray, width=\textwidth]
\textbf{Prompt Template:}\\
\textbf{Role:}\\
You are an advanced information system responsible for generating retrieval-oriented QA memory questions grounded in the provided atomic facts and original text.\\

\textbf{Task:}\\
Your task is to generate complex questions based on extracted atomic facts and the original text. The questions should be answerable using only the provided information and, when appropriate, require multi-fact integration (e.g., comparison, aggregation, or multi-hop reasoning) to support downstream retrieval and evidence-grounded answering.\\

\textbf{Requirements:}\\
Questions must strictly rely on the extracted atomic facts and original text, without introducing any external information.\\
Prefer questions that are specific, unambiguous, and informative for retrieval (avoid overly generic prompts).\\
Encourage compositional reasoning when supported by the facts (e.g., Set / Comparison / Aggregation / Multi-hop / Post-processing Heavy / False Premise).\\
Answers must accurately reflect the original content and refer to specific expressions in the text or atomic facts whenever possible.\\
Language must be clear and logically rigorous, avoiding ambiguity.\\

\textbf{Output Format (Follow this format strictly):}\\
\textbf{Example:}\\
\{example\}\\

Now, based on the following atomic facts and original paragraph, generate a complex question and its corresponding answer:\\
\textbf{Question Type:}\\
\{type\}\\
\textbf{Extracted Facts:}\\
\{extracted\_facts\}\\
\textbf{Original Text:}\\
\{text\}
\end{tcolorbox}
\caption{QA Memory Question Generation Prompt}\label{tab:qa-mem-prompt}
\end{table*}

\begin{table*}
\begin{tcolorbox}[colback=gray!2!white, colframe=gray, width=\textwidth]
\textbf{Prompt Template:}\\
\textbf{Role:}\\
You are now an intelligent assistant tasked with answering the final question based on the provided reference question-answer pairs and context documents. Follow these rules strictly: Only output the final answer, without any explanation or additional content.\\

\textbf{Reference Q\&A Pairs:} \{context\}\\
\textbf{Context Document:} \{document\}\\
\textbf{Question:} \{question\}\\
\textbf{Answer:}
\end{tcolorbox}
\caption{RAG Prompt}\label{tab:rag-prompt}
\end{table*}

\end{document}